\documentclass[11pt]{article}
\usepackage{graphicx} 
\usepackage{mathrsfs}  
\usepackage{fullpage}
\usepackage{amsthm}
\usepackage{amsmath}
\usepackage{mathtools}
\usepackage{amsfonts}
\usepackage{color}
\usepackage{physics}
\usepackage{amssymb}
\usepackage[colorlinks,linkcolor=blue,citecolor=blue,urlcolor=magenta,filecolor=cyan]{hyperref}
\usepackage[capitalize]{cleveref}
\usepackage{authblk}

\newtheorem{theorem}{Theorem}
\newtheorem{definition}[theorem]{Definition}
\newtheorem{proposition}[theorem]{Proposition}
\newtheorem{lemma}[theorem]{Lemma}
\newtheorem{corollary}[theorem]{Corollary}

\newcommand{\cH}{\mathcal H}
\newcommand{\cD}{\mathcal D}
\newcommand{\cV}{\mathcal V}
\newcommand{\cO}{\mathcal O}
\newcommand{\cL}{\mathcal L}
\newcommand{\cE}{\mathcal E}

\newcommand{\cB}{\mathcal B}

\newcommand{\cM}{\mathcal M}
\newcommand{\cT}{\mathcal T}
\newcommand{\cF}{\mathcal F}
\newcommand{\ran}{\operatorname{range}}

\title{Query-optimal quantum simulation of Lindblad evolution}

\author{Chunhao Wang}
\author{Christopher Ye}

\affil{Department of Computer Science and Engineering, Pennsylvania State University}
\affil[]{Email: \{cwang,cfy5119\}@psu.edu}

\date{}

\begin{document}

\maketitle

\begin{abstract}
For the problem of simulating Lindblad evolution for time $t$ to precision $\epsilon$, Hamiltonian simulation provides an additive query lower bound, informally, $\Omega(t + \mathrm{polylog}(1/\epsilon))$. However, the best previously known algorithms for general Lindblad simulation achieve a multiplicative upper bound, informally, $\mathcal{O}(t\,\mathrm{polylog}(1/\epsilon))$, in query complexity. It has remained open whether this multiplicative dependence is necessary. In this paper, we close the gap in query complexity by giving an algorithm with optimal additive dependence on evolution time and precision in the block-encoding model. Our approach uses the transducer framework to reduce the query cost of composing first-order approximations to the evolution channel, together with linear combinations of reuse circuits of different lengths to suppress catalyst-removal error. We further achieve nearly optimal gate complexity in evolution time and precision through history compression and an efficient implementation of the query-free part of the transducer using operation reordering and linear combinations of unitaries, while preserving the optimal query complexity.

\end{abstract}

\section{Introduction}

\paragraph{Motivation}

Simulating quantum systems was one of the original applications Feynman proposed for quantum computers in 1982 \cite{F82}, and it is likely to be among the \emph{first} applications of fault-tolerant quantum computers. Idealized quantum systems are isolated from their environment and governed by the Schr\"{o}dinger equation, whose generator is the Hamiltonian. Such systems are called \emph{closed quantum systems}, and their dynamics are referred to as Hamiltonian evolution.

Outside textbooks, virtually all realistic quantum systems interact with their environment to some degree. Consider, for example, the implementation of quantum gates. The need for quantum error correction already illustrates how difficult it is to realize idealized quantum dynamics in practice: interactions with the environment introduce noise and cause deviations from the intended evolution. It is therefore of considerable practical interest to study quantum systems that interact with their environment---\emph{open quantum systems}.

Among open quantum systems, an important class retains a useful feature of closed-system dynamics: the generator can be expressed entirely in terms of \emph{operators acting only on the system}. The influence of the environment is incorporated into these operators, without explicitly tracking the environmental degrees of freedom. This description applies when the environment can be treated as a memoryless reservoir: information transferred to the environment does not feed back into the subsequent system dynamics. Consequently, the state at $t + \delta$ is determined by the state at time $t$, without requiring knowledge of its earlier history. Systems with such memoryless dynamics are called \emph{Markovian} open quantum systems.

For time-independent dynamics, the Markovian property implies that the quantum channel $\cM_t$ describing the evolution satisfies $\cM_{t+s} = \cM_t \circ \cM_s$ and $\cM_0$ is the identity channel. In other words, evolving for time $s$ and then for time $t$ is equivalent to evolving for time $t+s$. Assuming continuity in time, these channels form a quantum dynamical semigroup. As shown in \cite{Lindblad1976,GKS76}, the generator of such a semigroup has the following form:
\begin{align}
  \label{eq:lindblad}
    \frac{\dd}{\dd t} \rho = \cL(\rho) = -i[H, \rho] + \sum_{j=1}^{m} (L_j \rho L_j^{\dagger} - \frac{1}{2} \{L_j^{\dagger}L_j, \rho\}),
\end{align}
which is called the \emph{Lindblad equation}, and the superoperator $\cL$ is called the \emph{Lindbladian}. Here $H$ is the Hamiltonian, and $L_j$'s are called the jump operators. The solution to \cref{eq:lindblad} is
\begin{align}
  \label{eq:lindblad-solution}
  \rho(t) = e^{\cL t}(\rho(0)),
\end{align}
where $e^{\cL t}$ is a quantum channel. The goal of simulating Lindblad evolution is to construct a quantum circuit that approximates the quantum channel $e^{\cL t}$, given efficient descriptions of the operators $H$ and $L_j$ in \cref{eq:lindblad}.

Although the Lindblad equation is a somewhat theoretical idealization, it is already capable enough to faithfully model many realistic open quantum systems, such as laser-driven atoms undergoing spontaneous emission~\cite{DCM92} and engineered dissipative dynamics in trapped-ion systems~\cite{BMS+11}. Beyond modeling physical systems, Lindblad evolution also serves as a computational primitive in quantum algorithms. Applications include Gibbs state preparation~\cite{RWW23,CKBG25}, continuous optimization through quantum Langevin dynamics~\cite{CLWLL25}, and thermal gradient descent toward approximate local energy minima~\cite{CHHPZ25}. Lindblad dynamics have also been proposed for preparing suitable initial states for quantum algorithms that sample non-logconcave distributions~\cite{LDCL26}. Furthermore, they provide routes to solving linear ordinary differential equations by encoding their solutions in off-diagonal blocks of density matrices~\cite {SGAZ25}, and to solving linear systems by encoding solutions in attractive stationary states~\cite{Shang26}. These applications motivate efficient Lindblad simulation as an important building block for sampling, optimization, and scientific computing.

An early efficient quantum algorithm for simulating Lindblad evolution was due to Kliesch et al.~\cite{K11} in 2011, based on a first-order Trotter decomposition. Suppressing the dependencies on other parameters, the number of local channels scales as \(\cO(t^2/\epsilon)\). In 2017, Childs and Li~\cite{CL17} introduced algorithms based on the Gram matrix and sparse Stinespring-isometry framework. Their algorithm for local jump operators uses \(\cO(t^{3/2}/\sqrt{\epsilon})\) gates, while their algorithm for a single sparse jump operator uses \(\cO(t^2/\epsilon)\) queries. Cleve and Wang introduced a version of linear combination of unitaries (LCU) for channels, achieving nearly optimal complexity with \(\cO(\tau \mathrm{ \ poly} \log (\tau/\epsilon))\) queries \cite{CW17}. Several later algorithms developed different approaches to Lindblad simulation. Li and Wang introduced an algorithm in 2023 which uses higher-order Duhamel expansions and LCU for channels \cite{LW23}. He et al.~\cite{HLLLW24} subsequently developed an algorithm for time-dependent Lindblad simulation with nearly linear dependence on the time-integrated block-encoding normalization and polylogarithmic dependence on inverse precision. In 2024, Ding, Li, and Lin constructed Hamiltonian evolutions on an enlarged system whose induced channels approximate Lindblad evolution to arbitrarily high order~\cite{DLL24}. Chen, Kastoryano, Brand{\~a}o, and Gily{\'e}n also developed a general Lindblad simulator with nearly linear dependence on time and polylogarithmic dependence on inverse precision, with no explicit dependence on the number of jump operators when their stacked operator is block encoded~\cite{CKBG25}. More recently, Chen, Li, Lu, and Ying developed a randomized approach that samples individual Lindblad terms, reducing the need to access all jump operators coherently~\cite{CLLY25}.

As a special case of the Lindblad simulation problem, Hamiltonian simulation provides the additive query lower bounds $\Omega(t + \mathrm{polylog}(1/\epsilon))$ that also apply to Lindblad simulation~\cite{BCK15,GSLW19} \footnote{This lower bound is actually of the form \(c(\tau, \epsilon)\) for \(c(\cdot)\) defined in \cref{eq:c-optimal-function} \cite[Corollary~60]{GSLW19}.}. However, unlike the Hamiltonian case, where optimal simulation algorithms saturate these lower bounds~\cite{LC17,GSLW19}, it has remained open since 2017 whether general Lindblad simulation can achieve an analogous additive dependence on evolution time and a polylogarithmic precision term. Recent efforts have narrowed this gap in restricted settings. Shang, An, and Shao achieved additive scaling for purely dissipative dynamics with unitary jump operators, with an extension to nonunitary jumps satisfying $\sum_j L_j^{\dag}L_j \propto I$~\cite{SAS26}. Under the same condition, Borras and Marvian developed a quantum-trajectory algorithm that accommodates a Hamiltonian term and achieves additive complexity for queries to jump operators, although the Hamiltonian-query complexity remains multiplicative \cite{BM26}. These advances motivate the question of whether additive query complexity can be achieved for general Lindbladians without restrictions on the total jump rate.

\paragraph{Main results}
In this paper, we resolve the query complexity gap by proposing a Lindblad simulation algorithm that achieves additive \emph{query complexity} using an approach based on the transducer framework, introduced in~\cite{BJY24}, and subsequently applied to time-dependent Hamiltonian simulation in~\cite{CGWZ26, CGWZ26gate}.

For the simulation problem, we consider a very general access model, namely, the block-encoding of operators. In particular, we assume access to a unitary \(O_H\) that is an \(\alpha_H\)-block encoding of the Hamiltonian $H$. For each \(j\in[m]\), we assume access to a unitary \(O_{L_j}\) that is an \(\alpha_j\)-block encoding of the \(j\)-th jump operator. We also assume controlled access to these oracles and their adjoints.

In some applications, such as quantum Gibbs sampling~\cite{CKBG25}, it is more convenient to consider the stacked jump operator $\widehat{L}$ as the input:
\begin{equation}
    \label{eq:normalized-stacked-jump}
    \widehat L \coloneqq
    \sum_{i=1}^{m} \frac{1}{\sqrt{\sum_{j=1}^m \alpha_j^2 }} \ket{i}\otimes L_i.
\end{equation} 
A unitary \(O_{\widehat{L}}\) that block encodes \(\widehat{L}\) can be constructed by first preparing the state \(\propto \sum_{j=1}^{m} \alpha_j \ket j\) and applying \(O_{L_i}\) controlled on the jump index register being equal to \(i\). To normalize the evolution time, we follow the convention of \cite{LW23} and define \(\norm{\cL}_{\mathrm{be}}\) as the block-encoding normalization factor, 
\begin{equation}
    \norm{\cL}_{\mathrm{be}} \coloneqq \alpha_H + \sum_{j=1}^m \alpha_j^2,
    \qquad
    \tau \coloneqq t\norm{\cL}_{\mathrm{be}}
    \label{eq:normalized-time}
\end{equation}
where \(\tau\) is the normalized time. We choose to omit the \(\frac{1}{2}\) coefficient on the dissipative normalization factor. Our main result is stated as follows.

\begin{theorem}[Informal version of \cref{thm:main-thm}]
\label{thm:main-thm-informal}
For a Lindbladian $\mathcal{L}$ as defined in \cref{eq:lindblad}, let \(t>0\) be the evolution time and \(0<\epsilon\leq1/2\) be the desired simulation error. There is a quantum algorithm that implements the quantum channel $e^{\mathcal{L}t}$ with diamond-norm error up to $\epsilon$, using $\cO(\tau + \mathrm{polylog}(\tau/\epsilon))$ queries to each of the block encodings of \(H, L_1, L_2, ..., L_m\) and their adjoints, where $\tau \coloneqq t \norm{\mathcal{L}}_{\mathrm{be}}$. The algorithm uses $\cO((1 + \tau  +  \mathrm{polylog}(\frac{1}{\epsilon}))(C_{\mathrm{anc}} + m + \mathrm{polylog}(\frac{m\tau}{\epsilon})))$ additional one- and two-qubit gates where \(C_{\mathrm{anc}}\) is the gate cost of the controlled operations on the block-encoding ancillas. 
\end{theorem}

Our results establish the optimal worst-case dependence on evolution time and precision in the query complexity of general Lindblad simulation. The precise bound in the formal version of the main theorem, \cref{thm:main-thm}, matches the corresponding lower bounds inherited from Hamiltonian simulation~\cite{BCK15,GSLW19}, and applies to both the Hamiltonian and jump-operator oracles, without requiring unitary jump operators or a state-independent total jump rate. Thus, the multiplicative dependence in previous general algorithms is not an intrinsic limitation of Lindblad simulation in the block-encoding model. Our construction extends the transducer approach to general dissipative dynamics and provides a framework for achieving this optimal scaling. Our gate complexity is near-optimal and comparable to that of other existing algorithms ~\cite{CW17,LW23,CKBG25}.

In addition, our algorithm also extends to simulating time-dependent Lindbladians, as outlined in \cref{sec:time-dep}.

\paragraph{Technical overview}

For Hamiltonian simulation, additive query complexity was achieved using quantum signal processing \cite{LC17}. This framework was later generalized by quantum singular value transformation, which applies polynomial transformations to the singular values of a block-encoded matrix. For Hermitian \(H\), suitable polynomial approximations to the complex exponential can be used to approximate \(e^{itH}\) \cite{GSLW19}. This approach does not directly extend to general Lindbladian dynamics because the generator \(\cL\) need not be normal and the simulation must implement \(e^{t\cL}\) as a quantum channel on input states. Approaches that independently approximate many short evolution steps typically incur multiplicative dependence on the normalized time \(\tau\) and the precision, since each of \(r\) steps must have error at most \(\epsilon/r\).

Inspired by the recent work achieving additive query complexity for time-dependent Hamiltonian simulation~\cite{CGWZ26}, we use a transducer-based approach to perform our simulation. Informally, a transducer is a unitary that transforms a public input while returning a private catalyst vector unchanged. The public input and the catalyst occupy orthogonal subspaces, and the oracle acts only on the private space. Transducers using the same oracle can be composed by assigning their catalysts to separate private sectors. The oracle then acts coherently on these sectors, while query-free unitaries carry out the successive transformations of the public state. 
This composition lets us separate the number of evolution steps from the number of queries and control the error over the entire evolution time. 

Our first contribution is a local dissipative transducer that implements a Stinespring isometry whose induced channel is a first-order approximation to dissipative evolution. We adapt the Cayley construction in~\cite{CGWZ26} by introducing distinct jump and no-jump output branches. These branches encode the Kraus operators of an exactly trace-preserving channel with local error $\cO(\delta^2)$ for a normalized time step $\delta$. The catalyst required for this local transformation has squared norm $\cO(\delta)$.

We interleave this local transducer with the Hamiltonian transducer from \cite{CGWZ26} and compose the resulting steps into a global transducer. An environment register records the jump history, while time and stage labels identify the separate private catalyst sectors. Each local rotation acts only on its corresponding private sector, so oracle actions on the other sectors commute past it. This allows the oracle calls to be collected into a single layer, using a constant number of controlled queries to each input oracle. The first-order discretization contributes error $\cO(\tau^2/r)$ after $r$ time steps, but increasing $r$ does not increase the query cost of the global transducer. Moreover, orthogonality of the private labels makes the squared catalyst norm add, giving a global catalyst map $\Gamma$ with $\norm{\Gamma}^2 = \cO(\tau)$. The number of query-free rotations still grows with $r$. 

The required catalyst depends on the input and the oracles and need not be efficiently preparable. To avoid preparing it, we use the reuse construction of~\cite{BJY24}. The public input is distributed uniformly over $s$ branches. By linearity, the catalyst required to transform each branch is decreased by a factor of $1/\sqrt{s}$. The same private catalyst can then be reused to transform the branches in succession, after which the public branches are recombined. Starting with zero private input introduces an error of at most $\norm{\Gamma}/\sqrt{s} = \cO(\sqrt{\tau/s})$ in the public output block.

The standard perturbation bound gives a query cost of $\cO(1+\tau/\epsilon^2)$~\cite{BJY24}. To achieve additive query complexity, we use the LCU over reuse lengths construction of Chen, Gao, Wang, and Zhou~\cite{CGWZ26}, which combines the public-output blocks of reuse circuits of different lengths. These blocks approximate the same global Stinespring isometry, but their errors depend on the reuse length. Suitable coefficients allow these errors to cancel while keeping the LCU normalization constant. One round of oblivious amplitude amplification gives the simulation channel with high fidelity. This construction reduces the catalyst-removal error, while the first-order discretization error is controlled separately by choosing $r$ sufficiently large. 

Our second contribution is the analysis showing that this LCU construction applies to dissipative evolution. The error depends on a polynomial evaluated on the global transducer's private-to-private block. Bounding this operator on the full private space includes transitions that cannot occur in chronological evolution. In particular, the local rotation contains a norm-one transition from the public jump-input branch to the private space. The fresh history register is initially $\ket{0}$, so this jump-input branch is empty when the corresponding stage begins. We therefore restrict the analysis to an invariant subspace of valid time and history labels that excludes this transition. Chronological ordering and the orthogonality of different jump histories then give factorial decay in the relevant polynomial degree. For a degree parameter $q$ sufficiently large compared with $1+\tau$, the resulting catalyst-removal error is bounded by
\begin{align*}
    \sqrt{2\tau}\left( C \sqrt{\frac{\tau}{q}}\right)^q
\end{align*}
for a universal constant $C$, using $\cO(q)$ queries to each oracle. Choosing $q$ to make this error sufficiently small yields the additive query bound, informally $\cO(\tau + \mathrm{polylog}(1/\epsilon))$, with the precise optimal dependence stated in~\cref{thm:main-thm}.

Our third contribution is a gate-efficient implementation of the transducer-based algorithm. A direct implementation has unfavorable scaling due to the large number of local channels required. To get more favorable scaling, we reorganize the implementation by performing all jump-producing rotations first. For each history, the remaining no-jump dissipative and Hamiltonian rotations form a suffix beginning after its last jump. We implement this suffix using a four-term LCU construction followed by exact oblivious amplitude amplification. In the reuse construction, the maximum number of jump events is bounded by the reuse length, which allows us to compress the history without changing the output of the reuse circuit. Combined with time segmentation, this allows us to achieve nearly optimal gate complexity. 

\paragraph{Concurrent work} Shortly after the first version of this paper was made public, independent work by Chen, Gao, Wang, and Zhou~\cite{CGWZ26a} appeared, achieving the same optimal query complexity and lower gate complexity than our original algorithm. At that time, we were writing up the independently developed gate-complexity improvements presented in this updated version. Both algorithms use the transducer framework and achieve optimal query complexity and similar nearly optimal gate complexity. We construct our global transducer by composing small-time-step transducers for the dissipative and Hamiltonian terms, whereas they use a local transducer that approximates the combined evolution. Our gate-efficient implementation uses LCU to implement query-free rotation suffixes associated with the no-jump dissipative and Hamiltonian steps, whereas theirs uses a rotation factorization over dyadic time intervals.

\paragraph{Organization}
 \cref{sec:prelim} specifies notation and definitions. \cref{sec:singe-jump} constructs an algorithm for one jump operator. \cref{sec:main-thm} extends the algorithm to multiple jump operators and a nonzero Hamiltonian and provides a gate-efficient implementation of the algorithm. \cref{sec:time-dep} briefly explains how to extend these results to the time-dependent setting. 

\section{Preliminaries}
\label{sec:prelim}

\subsection{Notation}
In this paper, we use \(\norm{\cdot}\) to denote the spectral norm and \(\norm{\cdot}_1\) to denote the trace norm. For a superoperator \(\cM : \cB(\cH_{\mathrm{in}}) \rightarrow \cB(\cH_{\mathrm {out}})\), where \(\cB(\cH)\) is the space of linear operators on \(\cH\), the induced trace norm \(\norm{\cM}_{1\rightarrow 1}\) is defined as \(\norm{\cM}_{1\rightarrow 1} = \max\{\norm{\cM(A)}_1: \norm{A}_1 = 1\}\). The diamond norm \(\norm{\cdot}_\diamond\) of a superoperator is defined as \(\norm{\cM}_\diamond = \norm{\cM \otimes \operatorname{id}_{\cH_{\mathrm{in}}}}_{1 \rightarrow 1}\) where \(\operatorname{id}_{\cH_{\mathrm{in}}}\) is the identity channel on a space isomorphic to \(\cH_{\mathrm{in}}\). We measure channel distance in diamond norm because it captures distinguishability. 

For Hermitian operators, we use \(\preceq\) as the positive-semidefinite ordering. That is, for \(A, B\), \(A \preceq B\) if \(B-A\) is positive-semidefinite.

We use the standard definition of a block encoding.

\begin{definition}[Block Encoding]
    Let \(A:\cH_{\mathrm{in}}\to\cH_{\mathrm{out}}\), where
    \(\cH_{\mathrm{in}}=(\mathbb C^2)^{\otimes\ell}\) and
    \(\cH_{\mathrm{out}}=(\mathbb C^2)^{\otimes k}\).
    For \(n\geq\max\{k,\ell\}\), let
    \(\cH_{\mathrm{enc}}=(\mathbb C^2)^{\otimes n}\)
    and define the zero-padding isometries
    \begin{equation}
        B_{\mathrm{in}}\ket\psi
        \coloneqq \ket{0^{n-\ell}}\otimes\ket\psi,
        \qquad
        B_{\mathrm{out}}\ket\phi
        \coloneqq \ket{0^{n-k}}\otimes\ket\phi.
        \label{eq:block-encoding-isometry}
    \end{equation}
    A unitary \(U\) on \(\cH_{\mathrm{enc}}\) is an
    \((\alpha, n,\epsilon)\)-block encoding of \(A\) if
    \begin{equation}
        \norm{A-\alpha B_{\mathrm{out}}^\dagger
        U B_{\mathrm{in}}}\leq\epsilon,
        \label{eq:block-encoding}
    \end{equation}
    where \(\alpha>0\) and \(\epsilon\geq0\).
\end{definition}

We call a unitary \(U\) an \(\alpha\)-block encoding of \(A\) if it is a \((\alpha, n, 0)\)-block encoding of \(A\) for some \(n > 0\).

\subsection{The infinitesimal approximation map}
In this subsection we define the channel that our local transducer implements and show that it is a first-order accurate approximation to dissipative evolution for some short time step \(0 \leq \delta \leq \frac{1}{2}\). We first show the case of a single jump operator and then show the extension to multiple jump operators using our stacked construction from \cref{eq:normalized-stacked-jump}. 

For a jump operator \(L\), define the associated dissipator by
\begin{equation}
    \cD_L(\rho) \coloneqq L\rho L^\dagger-\frac{1}{2}\{L^\dagger L,\rho\}.
    \label{eq:dissipative-lindblad}
\end{equation}
The full dissipative portion of the Lindblad equation is \(\sum_{\ell=1}^{m} \cD_{L_\ell}\).

Let \(O_L\) be an \(\alpha\)-block encoding of \(L\).
Define a scalar \(\nu\) as follows. 
\begin{equation}
    \nu \coloneqq \frac{1-\sqrt{1-\delta}}{1+\sqrt{1-\delta}}
    \label{eq:nu}
\end{equation}
The channel \(\Phi_\delta(\rho)\) is
\begin{equation}
    \Phi_\delta(\rho) \coloneqq  K_0\rho K_0^\dagger + K_1\rho K_1^\dagger
    \label{eq:local-channel}
\end{equation} where
\begin{equation}
    K_0 = (I- \nu \frac{L^\dagger L}{\alpha^2})(I+\nu \frac{L^\dagger L}{\alpha^2})^{-1},
    \qquad
    K_1 =-2i\sqrt{\nu}\,\frac{L}{\alpha}(I+\nu \frac{L^\dagger L}{\alpha^2})^{-1}.
    \label{eq:local-kraus}
\end{equation}

\(\Phi_{\delta}(\rho)\) is a quantum channel because 
\begin{align}
    K_0^\dagger K_0 + K_1^\dagger K_1
    &=(I+\nu \frac{L^\dagger L}{\alpha^2})^{-2}\bigl((I-\nu \frac{L^\dagger L}{\alpha^2})^2+4\nu \frac{L^\dagger L}{\alpha^2}\bigr)\\
    &=(I+\nu \frac{L^\dagger L}{\alpha^2})^{-2}(I+\nu\frac{L^\dagger L}{\alpha^2})^2=I.
\end{align}

\begin{lemma}[Uniform first-order error]
\label{lem:local-error}
For \(0\leq\delta\leq \frac{1}{2}\), there is a universal constant
\(C_{\mathrm{loc}}>0\) such that
\begin{equation}
    \norm{\Phi_\delta - \exp(\frac{\delta}{\alpha^2}\cD_L)}_\diamond \leq C_{\mathrm{loc}}\delta^2,
    \label{eq:local-error}
\end{equation}
\end{lemma}
\begin{proof}
First note that
\begin{align}
    \nu = \frac{\delta}{(1+\sqrt{1-\delta})^2} = \frac{\delta}{4} + \cO(\delta^2).
\end{align}
Since \(0\preceq \frac{L^{\dagger}L}{\alpha^2}\preceq I\), the Taylor expansions hold uniformly for \(0\leq\delta\leq1/2\) in operator norm:
\begin{align*}
    (I + \nu \frac{L^{\dagger}L}{\alpha^2})^{-1} & =I - \nu \frac{L^{\dagger}L}{\alpha^2} + \cO(\nu^2).
\end{align*}
Substituting gives
\begin{align*}
    K_0
    &= (I-\nu \frac{L^{\dagger}L}{\alpha^2})(I+\nu \frac{L^{\dagger}L}{\alpha^2})^{-1}\\
    &= (I-\nu \frac{L^{\dagger}L}{\alpha^2})\bigl(I-\nu \frac{L^{\dagger}L}{\alpha^2}+\cO(\nu^2)\bigr)\\
    &= I-2\nu \frac{L^{\dagger}L}{\alpha^2}+\cO(\nu^2)\\
    &= I-\frac{\delta}{2}\frac{L^{\dagger}L}{\alpha^2}+\cO(\delta^2),
\end{align*}
and
\begin{align*}
    K_1
    &= -2i\sqrt{\nu}\,\frac{L}{\alpha}(I+\nu \frac{L^{\dagger}L}{\alpha^2})^{-1}\\
    &= -i\bigl(\sqrt{\delta} + \cO(\delta^{3/2})\bigr) \frac{L}{\alpha}\bigl(I - \nu \frac{L^{\dagger}L}{\alpha^2} + \cO(\nu^2)\bigr)\\
    &= -i\sqrt{\delta}\,\frac{L}{\alpha}+\cO(\delta^{3/2}).
\end{align*}
Thus
\begin{align*}
    \Phi_\delta(\rho) & =\rho + \frac{\delta}{\alpha^2}\left( L\rho L^\dagger-\frac{1}{2} L^{\dagger}L \rho-\frac{1}{2}\rho L^{\dagger}L \right) + \cO(\delta^2).
\end{align*} 

Comparing to the Taylor expansion of \(e^{\frac{\delta}{\alpha^2}\cD_L}\) and using \(\norm{\rho\mapsto M\rho N}_\diamond\leq\norm {M}\norm {N}\) to bound the remainder in the diamond norm proves the claim.  
\end{proof}

\begin{corollary}[Time discretization]
\label{cor:discretization}
Let \(r\in\mathbb N, r\geq 2 \alpha^2t\), set
\(\delta=\alpha^2 t /r\), and apply the local channel \(r\) times. Then
\begin{equation}
    \norm{\Phi_\delta^r-\exp(t\cD_{L})}_\diamond
    \leq C_{\mathrm{loc}}\frac{\alpha^4t^2}{r}.
    \label{eq:global-discretization}
\end{equation}
\end{corollary}
\begin{proof}
Both \(\Phi_\delta\) and \(\exp(\frac{\delta}{\alpha^2}\cD_L)\) have diamond norm one. A telescoping sum and Lemma~\ref{lem:local-error} give \(r C_{\mathrm {loc}}\delta^2=C_{\mathrm{loc}}\alpha^4t^2/r\).
\end{proof}

To handle multiple jump operators, we use the stacked Kraus-operator formula in \cref{eq:local-kraus} with \(L=\widehat L\) and \(\alpha=1\). Let \(\mathrm J\) denote the jump-index register.

Define \(\Lambda_\mathrm{D} = \sum_{i=1}^{m} \alpha_i^2\) and let
\[
    \Phi_{\delta}(\rho)
    =K_0\rho K_0^\dagger+
    \Tr_{\mathrm{J}} (K_1\rho K_1^\dagger).
\]
Define 
\begin{equation}
    K_0=(I-\nu\widehat{L}^{\dagger}\widehat{L})(I+\nu\widehat{L}^{\dagger}\widehat{L})^{-1},
    \qquad
    K_{\mathrm{J}, \ell}=-2i\sqrt \nu\,\frac{L_{\ell}}{\sqrt{\Lambda_{\mathrm{D}}}}(I+\nu\widehat{L}^{\dagger}\widehat{L})^{-1}.
    \label{eq:stacked-kraus}
\end{equation}
\(\Phi_{\delta}\) is a valid quantum channel because 
\begin{equation}
    K_0^\dagger K_0+\sum_{\ell=1}^m K_{\mathrm{J}, \ell}^\dagger K_{\mathrm{J}, \ell}
    = (I+\nu\widehat{L}^{\dagger}\widehat{L})^{-2} \bigl((I-\nu\widehat{L}^{\dagger}\widehat{L})^2+4\nu\widehat{L}^{\dagger}\widehat{L}\bigr) = I.
    \label{eq:stacked-completeness}
\end{equation}
Substituting these expansions and using contractivity under partial trace gives
\begin{equation}
    \Phi_{\delta}(\rho) = K_0 \rho K_0^{\dagger} + \sum_{\ell = 1}^{m} K_{\mathrm{J}, \ell} \rho K_{\mathrm{J}, \ell}^{\dagger}
    =\rho+\frac{\delta}{\Lambda_{\mathrm{D}}} \sum_{\ell=1}^{m}\bigl(L_{\ell}\rho L_{\ell}^{\dagger} - \frac{1}{2}\{L_{\ell}^{\dagger}L_{\ell}, \rho\} \bigr) 
       +\cO(\delta^2)
    \label{eq:stacked-first-order-error}
\end{equation} where the error is in diamond norm.

\subsection{Transducers}

A transducer is a unitary \(S\) that, given an input \(\ket \psi\) and a catalyst vector \(\Gamma\ket \psi\), applies a target unitary \(U\) to the input vector and returns the catalyst vector unchanged. Following Belovs, Jeffery, and Yolcu, we define a public space \(\cH_{\mathrm P}\) and a private space \(\cH_{\mathrm C}\). A transducer acts on the direct sum of these two spaces, \(\cH_{\mathrm P} \oplus \cH_{\mathrm C}\). The target unitary \(U\) transforms the public input \(\ket \psi\) while the private space containing the catalyst vector is returned unchanged. The catalyst vector and combined vector need not be normalized. 

More formally, we use the transducer definition by Belovs, Jeffery, and Yolcu \cite[Section 5.1]{BJY24}. 
\begin{definition}[Transducer]
    \label{def:transducer}
    Let \(U\) be a unitary on \(\cH_{\mathrm P}\). Let \(S\) be a unitary on \(\cH_{\mathrm P} \oplus \cH_{\mathrm C}\). \(S\) is a transducer that implements \(U\) if there exists a linear map \(\Gamma: \cH_{\mathrm P} \rightarrow \cH_{\mathrm C}\) such that for all \(\ket\psi \in \cH_{\mathrm P}\)

    \begin{equation}
        S (\ket \psi \oplus \Gamma\ket \psi) = U\ket\psi \oplus \Gamma \ket \psi.
    \end{equation} 
\end{definition}

For transducers with access to an oracle, we use the following definition adapted from Belovs, Jeffery, and Yolcu \cite[Section 7.1]{BJY24}.
\begin{definition}[Canonical Transducer]
    \label{def:canonical-transducer}
    Let \(O\) be an oracle unitary acting on the private space \(\cH_{\mathrm C}\). Let \(\widehat{O} \coloneqq (I_{\mathrm P} \oplus O)\). \(S(O)\) is a canonical transducer that implements an oracle-dependent unitary \(U_O\) if there exists an oracle-independent unitary \(S^{\circ}\) on \(\cH_{\mathrm P} \oplus \cH_{\mathrm C}\) and an oracle-dependent linear map \(\Gamma_O: \cH_{\mathrm P} \rightarrow \cH_{\mathrm C}\) such that 
    \[
        S(O) \coloneqq S^{\circ}\widehat O
    \]
    and, for every \(\ket\psi \in \cH_{\mathrm P}\)
    \begin{equation}
        S^{\circ} \widehat{O}(\ket \psi \oplus \Gamma_O \ket \psi) = S^{\circ} (\ket\psi \oplus O \ \Gamma_O\ket \psi) = U_O\ket\psi \oplus \Gamma_O \ket\psi.
    \end{equation}
\end{definition}

Our construction uses a version of the canonical transducer with its inputs restricted to a subspace of \(\cH_{\mathrm P}\). Thus, we introduce an input-restricted version of the canonical transducer.
\begin{definition}[Input-restricted canonical transducer]
\label{def:input-restricted-canonical-transducer}
    Let \(O\) be an oracle unitary acting on the private space \(\cH_\mathrm{C}\), let \(V_{\mathrm{in}}:\cH_{\mathrm{in}} \to \cH_{\mathrm{P}}\) be an oracle-independent isometry, and let \(V_O:\cH_{\mathrm{in}} \to \cH_{\mathrm{P}}\) be an oracle-dependent target isometry. A unitary \(S(O)\) on \(\cH_{\mathrm{P}} \oplus \cH_{\mathrm{C}}\) is an input-restricted canonical transducer if
    \[
        S(O)=S^\circ(I_\mathrm{P}\oplus O),
    \]
    where \(S^\circ\) is an oracle-independent unitary, and there is a linear map \(\Gamma_O:\cH_{\mathrm{in}}\to\cH_{\mathrm{C}}\) such that
    \[
        S(O)\bigl(V_{\mathrm{in}}\ket\psi\oplus\Gamma_O\ket\psi\bigr)
        =V_O\ket\psi\oplus\Gamma_O\ket\psi
    \]
    for every \(\ket\psi\in\cH_{\mathrm{in}}\).
\end{definition}

To avoid the need to prepare this catalyst vector, we use the standard reuse construction of Belovs, Jeffery, and Yolcu \cite[Theorem 3.2]{BJY24}. We briefly restate this construction for convenience. 

For some \(s \in \mathbb{N}\) to be chosen later, define the ancilla space \(\cH_{\mathrm {reuse}}=\operatorname{span}\{\ket 0,\ldots,\ket{s-1}\}\) with dimension \(s\). The reuse construction acts on \(\bigl(\cH_\mathrm{reuse} \otimes \cH_{\mathrm {P}}\bigr) \oplus \cH_{\mathrm{C}}\). We define the combined ideal input vector as 
\[
    (\ket 0_{\mathrm{reuse}} V_{\mathrm{in}} \ket \psi)_\mathrm{P} \oplus \frac{1}{\sqrt{s}}(\Gamma_O\ket \psi)_\mathrm{C}.
\]

Applying the quantum Fourier transform \(F_s\) to the reuse register on the public subspace while acting as the identity on the private subspace gives
\begin{align}
    \left[(F_s \otimes I) (\ket 0_{\mathrm{reuse}} V_{\mathrm{in}} \ket \psi)_\mathrm{P} \right ] \oplus \frac{1}{\sqrt{s}}(\Gamma_O\ket \psi)_\mathrm{C} =
    \bigl( \frac{1}{\sqrt{s}}\sum_{k=0}^{s-1} \ket k_{\mathrm{reuse}} \otimes V_{\mathrm{in}} \ket \psi \bigr)_\mathrm{P} \oplus \bigl(\frac{1}{\sqrt{s}}\Gamma_O\ket\psi\bigr)_{\mathrm{C}}.
\end{align} 

The Fourier transform coherently distributes the public input equally among \(s\) branches. By linearity, the catalyst required for each branch is scaled by \(\frac{1}{\sqrt{s}}\).

For \(k\in \{0, 1,  \ldots, s-1\}\), define the isometry
\[
    J_k:\cH_{\mathrm P}\oplus\cH_{\mathrm C} \rightarrow  \bigl(\cH_\mathrm{reuse} \otimes \cH_{\mathrm {P}}\bigr) \oplus \cH_{\mathrm{C}}
\]
by
\begin{equation}
    J_k(\ket p_\mathrm{P}\oplus\ket c_{\mathrm{C}}) 
    =
    \bigl(\ket k_{\mathrm{reuse}} \otimes \ket p_{\mathrm{P}} \bigr) \oplus \ket c_{\mathrm{C}}.
\end{equation}

Define the \(k^{\mathrm{th}}\) stage transducer by
\begin{equation}
    S^{(k)}(O) \coloneqq J_k S(O) J_k^{\dagger} + \bigl(I - J_kJ_k^{\dagger}).
\end{equation}

\(S^{(k)}(O)\) acts as \(S(O)\) on \(\bigl(\ket k_\mathrm{reuse} \otimes \cH_{\mathrm{P}} \bigr) \oplus \cH_{\mathrm{C}}\) and as the identity on all other public branches. The private space is shared among all \(S^{(k)}(O)\) so the same catalyst can be reused to transform each part of the public state.

The input-restricted canonical transducer identity gives
\[
     S^{(k)}(O) \left[\bigl(\ket k_\mathrm{reuse} \otimes \frac{1}{\sqrt{s}}V_{\mathrm{in}} \ket \psi\bigr) \oplus \frac{1}{\sqrt{s}}\Gamma_O \ket \psi \right] 
     = 
     \bigl(\ket k_\mathrm{reuse} \otimes \frac{1}{\sqrt{s}} V_O \ket \psi\bigr) \oplus \frac{1}{\sqrt{s}} \Gamma_O \ket \psi,
\] and thus
\[
    S^{(s-1)}(O) \cdots S^{(0)}(O)
    \left[
        \bigl(\frac1{\sqrt s}\sum_{k=0}^{s-1}
            \ket k_\mathrm{reuse}\otimes V_{\mathrm{in}}\ket\psi\bigr)
        \oplus\frac{\Gamma_O\ket\psi}{\sqrt s}
    \right]
    =
    \bigl(\frac1{\sqrt s}\sum_{k=0}^{s-1}
        \ket k_\mathrm{reuse}\otimes V_O\ket\psi\bigr)
    \oplus\frac{\Gamma_O\ket\psi}{\sqrt s}.
\]
The public branches can then be recombined by applying \(F_s^{\dagger}\) to the reuse register yielding \(\bigl(\ket 0_\mathrm{reuse} \otimes V_O \ket \psi \bigr)_{\mathrm{P}} \oplus \frac{1}{\sqrt{s}} (\Gamma_O \ket \psi)_{\mathrm{C}}\). 

Omitting this scaled catalyst changes the input vector by at most \(\norm{\Gamma_O}/\sqrt{s}\) and makes the vector a valid quantum state. Let \(R_s\) denote the length-\(s\) catalyst-free reuse unitary including the initial and final Fourier transforms. The catalyst-free circuit applies \(R_s\) with zero private input. 

Following \cite{CGWZ26} we define \(P_s\) as the block of \(R_s\) where the input and output are restricted to the public component with reuse ancilla in the \(\ket0_{\mathrm{reuse}}\) state. For \(\ket\phi \in \cH_{\mathrm{P}}\) define
\begin{equation}
    B_{P_s, \mathrm{in}}\ket\phi = B_{P_s, \mathrm {out}}\ket\phi \coloneqq \bigl(\ket 0_\mathrm{reuse} \otimes \ket \phi\bigr) \oplus 0_\mathrm{C}
    \qquad
    P_s \coloneqq B_{P_s, \mathrm{out}}^{\dagger} R_s B_{P_s, \mathrm{in}}.
\end{equation}

Unitaries preserve distance, and compression onto the public output removes the ideal catalyst component and cannot increase the error. Thus,
\begin{equation}
    \norm{V_O-P_sV_{\mathrm{in}}}
    \leq
    \frac{\norm{\Gamma_O}}{\sqrt s}.
\end{equation}

To obtain error at most \(\epsilon\), set
\begin{equation}
    s = \Theta (\max_{O}{\bigl(1+\frac{||\Gamma_O||^2}{\epsilon^2}\bigr)}).
\end{equation}

This construction uses \(s\) applications of \(S(O)\) and each \(S(O)\) uses one query.

We will henceforth omit the oracle dependence and set \(V \coloneqq V_O\), \(\Gamma \coloneqq \Gamma_O\), \(S \coloneqq S(O)\) and note that the results hold uniformly over all admissible oracles. 

A more detailed analysis of the error in this reuse construction is performed in \cite{CGWZ26}. We restate the main argument and adapt it for input-restricted transducers.

Decompose \(S\) as follows:
\begin{equation}
    S = \begin{pmatrix}
        S_{\mathrm{P} \leftarrow \mathrm{P}} & S_{\mathrm{P} \leftarrow \mathrm{C}} \\ 
        S_{\mathrm{C} \leftarrow \mathrm{P}} & S_{\mathrm{C} \leftarrow \mathrm{C}}
    \end{pmatrix},
    \label{eq:s-decomp}
\end{equation}
where the first row and column correspond to the public space and the second to the private space. 

For some input \(\ket\psi\), let \(\ket{c_j}\) be the private vector entering the \(j^{\mathrm{th}}\) step of the reuse circuit and \(\ket{p_j} = V_{\mathrm{in}}\ket\psi/\sqrt{s}\) be the public state entering the \(j^{\mathrm{th}}\) step of the reuse circuit. The private vector after applying \(S\) in the \(j^{\mathrm{th}}\) step is \(S_{\mathrm{C} \leftarrow \mathrm{P}} \ket{p_j} + S_{\mathrm{C} \leftarrow \mathrm{C}} \ket{c_j}\). The fresh public states in the ideal and catalyst-free circuits are identical. Thus, the difference between the ideal and realized private states is determined by the private-to-private block multiplied by the error in the previous private state. 

The error in the input private state is \(\Gamma \ket \psi/\sqrt{s}\) because we use an empty catalyst in the catalyst-free reuse circuit. It follows by induction that the overall error in the public output after applying the inverse Fourier transform and projecting onto the public component with reuse ancilla \(\ket0_\mathrm{reuse}\) is
\begin{equation}
    \sum_{j=0}^{s-1} \frac{1}{s} \bigl (S_{\mathrm{P} \leftarrow \mathrm{C}} S_{\mathrm{C} \leftarrow \mathrm{C}}^j \Gamma \ket \psi \bigr).
\end{equation}

Define 
\begin{equation}
    g_s(A) \coloneqq \frac{1}{s} \sum_{j=0}^{s-1}A^j. 
    \label{eq:poly-g}
\end{equation}
Since the error identity holds for all \(\ket \psi\), 
\begin{equation}
    V - P_sV_{\mathrm{in}} = \frac{1}{s}\sum_{j=0}^{s-1}S_{\mathrm{P} \leftarrow \mathrm{C}} S_{\mathrm{C}\leftarrow \mathrm{C}}^j \Gamma = S_{\mathrm{P}\leftarrow \mathrm{C}} g_s(S_{\mathrm{C}\leftarrow \mathrm{C}}) \Gamma.
    \label{eq:error-poly}
\end{equation}

Following \cite[Sections~2 and~7.2]{CGWZ26}, we show how to combine reuse block encodings of different lengths using LCU to reduce error. 

\begin{lemma}[LCU error reduction for Input-Restricted Canonical Transducers]
    \label{lemma:lcu-error-reduction}
    
    Let $S$ be an input-restricted canonical transducer implementing an isometry \(V:\cH_{\mathrm{in}}\to\cH_{\mathrm{P}}\), with source embedding \(V_{\mathrm{in}}\).
    
    For \(\lambda_1, \lambda_2, ..., \lambda_{s_{\mathrm max}} \in \mathbb{C}\) define
    \begin{equation}
        \widetilde{P} \coloneqq \sum_{j=1}^{s_{\mathrm{max}}} \lambda_j P_j, 
        \qquad 
        \widetilde{V} \coloneqq \widetilde{P}V_{\mathrm in},
        \qquad
        p_{\mathrm{err}}(A) \coloneqq \sum_{j=1}^{s_{\mathrm{max}}} \lambda_j g_j(A), 
        \qquad 
        \Lambda_{\mathrm{LCU}} \coloneqq \sum_{i=1}^{s_{\mathrm{max}}} |\lambda_i|.
        \label{lcu-reuse-poly}
    \end{equation} 
    If \(\sum_{j=1}^{s_{\mathrm{max}}} \lambda_j = 1\), then a block encoding of \(\widetilde{P}\) with normalization \(\Lambda_{\mathrm{LCU}}\) can be constructed using \(s_{\mathrm{max}}\) controlled applications of \(S\) through a layerwise SELECT construction. Restricting its input through \(V_{\mathrm{in}}\) gives a block encoding of \(\widetilde V\) with the same normalization. Additionally

    \begin{equation}
        V-\widetilde{V} =  S_{\mathrm{P} \leftarrow \mathrm{C}} p_{\mathrm{err}}(S_{\mathrm{C}\leftarrow \mathrm{C}}) \Gamma.
        \label{eq:reuse-error-poly}
    \end{equation}
\end{lemma}

\section{Simulating a Single Jump Operator}
\label{sec:singe-jump}
In this section, we consider the restricted setting where \(H = 0\) and \(m=1\). We assume access to an oracle \(O_L\) which is an \(\alpha\)-block encoding of the jump operator \(L\). \(\norm{L/\alpha} \leq 1\).

The normalized time is \(\tau \coloneqq t\alpha^2.\)

We first construct a transducer that, given a purification \(\ket\psi\) of \(\rho\), a catalyst, and a time step \(\delta\), implements a Stinespring isometry for a channel \(\Phi_\delta\) satisfying
\[
    \|\Phi_\delta-e^{\delta\cD_{L/\alpha}}\|_\diamond=\cO(\delta^2).
\] 
We call this transducer the local transducer. We compose \(r\) steps of the local transducer to create the global transducer. We then bound the norm of a polynomial evaluated on the private-to-private block of the global transducer. This allows us to use Lemma~\ref{lemma:lcu-error-reduction} to reduce the error using the LCU construction of \cite{CGWZ26}.

For \(\tau>0\) and \(0 < \epsilon \leq \frac{1}{2}\), define
\begin{equation}
    c(\tau,\epsilon) \coloneqq 1 +  \tau + 
    \frac{\log(1/\epsilon)} 
    {\log\!\left(e+\dfrac{\log(1/\epsilon)}{\tau}\right)},
    \qquad
    c(0, \epsilon) \coloneqq 1.
    \label{eq:c-optimal-function}
\end{equation} This function describes the dependence on normalized time and precision for our simulation algorithms. 

\begin{theorem}[Simulation of a Single Jump Operator]
\label{thm:single-jump}
Let \(t\geq0\) and \(0<\epsilon\leq1/2\). Given controlled access to \(O_L\) and \(O_L^\dagger\), there is a quantum channel \(\widetilde{\Phi}\) such that 
\begin{equation}
    \norm{\widetilde{\Phi}-
    \exp\bigl(t\cD_{L}\bigr)}_\diamond =
    \norm{\widetilde\Phi-
    \exp\bigl(\tau \cD_{L/\alpha}\bigr)}_\diamond
    \leq\epsilon.
\end{equation}

The channel uses
\begin{equation}
    \cO\left(c(\tau,\epsilon)\right)
    \label{eq:query-complexity}
\end{equation}
queries to \(O_L\) or \(O_L^\dagger\). 
\end{theorem}

\subsection{Local Transducer}
We aim to construct a transducer that implements a Stinespring isometry whose induced channel approximates dissipative evolution to first order in the time step \(0<\delta\leq1/2\). Let \(\cH_{\mathrm{E}}\) denote the existing environment space. The input \(\ket\psi\) purifies \(\rho\) such that \(\Tr_{\mathrm{E}}{\ketbra{\psi}} = \rho\). We suppress the environment register below as all operators act as the identity on it.

To define the public and private spaces we introduce one ancilla qubit \(\cH_{\mathrm A}\coloneqq\mathbb C^2\) and let
\begin{equation}
    \widetilde L \coloneqq \ket 0 \bra 1 \otimes L^{\dagger} + \ket 1 \bra 0 \otimes L = \begin{pmatrix} 
        0 & L^\dagger \\ 
        L & 0
    \end{pmatrix}
\end{equation}
and define the isometries
\begin{equation}
    V_0\coloneqq\ket0_\mathrm{A}\otimes {B_{\mathrm{in}}},
    \qquad
    V_1\coloneqq\ket1_\mathrm{A}\otimes {B_{\mathrm{out}}},
    \qquad
    V_{\mathrm{enc}}\coloneqq(V_0,V_1),
    \label{eq:V-local}
\end{equation} 
where \({B_{\mathrm{in}}}, {B_{\mathrm{out}}}\) are the block-encoding isometries as defined in \cref{eq:block-encoding-isometry}.

The public and private spaces are
\[
\underbrace{\cH_{\mathrm{E}}}_{\text{shared environment}} 
\otimes
\left[
\underbrace{\begin{matrix}
    \underbrace{\cH_{\mathrm{S}}^{(0)}}_{\text{no-jump}} \oplus \underbrace{\cH_{\mathrm{S}}^{(1)}}_{\text{jump}}
\end{matrix}}_{\text{public}}
\;\oplus\;
\underbrace{\begin{matrix}
    \cH_{\mathrm A} \otimes \cH_{\mathrm{enc}}
\end{matrix}}_{\text{private}}
\right],
\]
\begin{equation}
    \cH_{\mathrm{P}, \mathrm{loc}} \coloneqq (\cH_{\mathrm{S}}^{(0)} \oplus \cH_{\mathrm{S}}^{(1)}) 
    \qquad
    \cH_{\mathrm{C}, \mathrm{loc}} \coloneqq \cH_{\mathrm{A}} \otimes\cH_{\mathrm {enc}},
    \label{eq:local-spaces}
\end{equation} 
where \(\cH_\mathrm{S} \cong \cH_{\mathrm{S}}^{(0)} \cong \cH_{\mathrm{S}}^{(1)}\) and superscripts only label the two copies as no jump and jump, respectively. The two spaces hold the two branches of the Stinespring isometry. Define \(V_{\mathrm{in}, \mathrm{loc}}:\cH_{S} \rightarrow (\cH_{\mathrm{S}}^{(0)} \oplus \cH_{\mathrm{S}}^{(1)})\) by \(V_{\mathrm{in}, \mathrm{loc}}\ket \psi_{\mathrm{S}} = \ket \psi_{\mathrm{S}^{(0)}} \oplus \vec{0}_{\mathrm{S}^{(1)}}\).

Define the Hermitian embedding 
\begin{equation}
  \widetilde{O}_L \coloneqq
    \begin{pmatrix}
        0 & O_L^\dagger \\ 
        O_L & 0
    \end{pmatrix},
  \label{eq:hermitian-oracle}
\end{equation}

then
\begin{equation}
    \widetilde{O}_L^\dagger=\widetilde{O}_L,
    \qquad \widetilde{O}_L^2=I,
    \qquad V_{\mathrm{enc}}^\dagger\widetilde{O}_LV_{\mathrm{enc}}=\frac{\widetilde L}{\alpha},
    \label{eq:lift-compression}
\end{equation}
where \(\widetilde{O}_L\) is the oracle for the local step.

When the inverse exists, define the complex Cayley transform
\begin{equation}
    C(A) \coloneqq (I-iA)(I+iA)^{-1}
\end{equation}
and the real Cayley transform
\begin{equation}
    F(A) \coloneqq C(-iA) = (I-A)(I+A)^{-1}.
\end{equation}

For a public state \(\ket \psi\) we construct our catalyst \(\ket x\) using two intermediate vectors, \(\ket y, \ket z\). Following the Cayley-transducer construction of \cite{CGWZ26} we design \(\ket z, \ket x\)  so that a query on \(\ket x\) produces an encoding of 
\(iC(\frac{\widetilde{L}}{\alpha}) (\ket 0_A \ket y)\) in the \(V_{\mathrm{enc}}\) subspace. In the orthogonal subspace, the oracle acts as a phase. The vector \(\ket y\) is chosen to allow us to create the two public outputs as well as restore the catalyst with the amplitude supplied by the public input. 
\begin{align}
    \ket y& \coloneqq \sqrt{\delta} \bigl(I+\sqrt{1-\delta} F(L^\dagger L /\alpha^2) \bigr)^{-1}\ket \psi \label{eq:caty} \\
    \ket z& \coloneqq \bigl(I+i\frac{\widetilde{L}}{\alpha}\bigr)^{-1} (\ket 0_A \ket y) \label{eq:catz} \\ 
    \ket x& \coloneqq \bigl(I+i\widetilde{O}_L \bigr)V_{\mathrm{enc}}\ket z. \label{eq:catx}
\end{align}
Since \(0\preceq F(L^\dagger L/\alpha^2)\preceq I\) and \(\widetilde L\) is Hermitian, the inverses in these definitions exist.

Let \(\Gamma_{\mathrm{loc}}:\cH_{\mathrm{S}}\to\cH_{\mathrm{C}, \mathrm loc}\) denote the linear catalyst map
\[
    \Gamma_{\mathrm{loc}}\ket\psi \coloneqq \ket{x(\psi)}.
\]

Define the rotation \(R\) 
\begin{equation}
  R \coloneqq
    \begin{pmatrix}
        \sqrt{1-\delta}I & 0 & -i\sqrt{\delta}V_0^\dagger \\
        0 & 0 & -iV_1^\dagger \\
        \sqrt{\delta}V_0 & V_1 & i(\sqrt{1-\delta}V_0V_0^{\dagger}+\bigl(I-V_0V_0^{\dagger}-V_1V_1^{\dagger} \bigr))
    \end{pmatrix}.
    \label{eq:R}
\end{equation}

\(R\) acts as follows.
\begin{equation}
    \begin{pmatrix}
        \text{public no-jump output} \\ 
        \text{public jump output} \\ 
        \text{private output} 
    \end{pmatrix}
    = R \begin{pmatrix}
        \text{public input} \\ 
        \vec{0} \\ 
        \text{private input}
    \end{pmatrix}.
\end{equation}

Define the local transducer \(S_{\mathrm{loc}}\) as
\begin{equation}
    S_{\mathrm{loc}} \coloneqq R (I_{\cH_{\mathrm{P}, \mathrm{loc}}}\oplus \widetilde{O}_L).
    \label{eq:local-transducer}
\end{equation}

\begin{lemma}[\(R\) is unitary]
    The operator \(R\) in \cref{eq:R} is unitary.
\end{lemma}
\begin{proof}
    Decompose the space into \(3\) orthogonal invariant subspaces: the direct sum of the public no-jump space and \(\ran V_0\), the direct sum of the public jump space and \(\ran V_1\), and the orthogonal complement of \(\ran V_0 \oplus \ran V_1\) in the private space. Under the identifications induced by \(V_0, V_1\), the restrictions of \(R\) are
    \[
    \begin{pmatrix}
        \sqrt{1-\delta} & -i\sqrt{\delta}   \\ 
        \sqrt{\delta} & i\sqrt{1-\delta}
    \end{pmatrix}, \qquad
    \begin{pmatrix}
        0 & -i \\ 
        1 & 0
    \end{pmatrix}, \qquad
    iI_{\ran \bigl(I-V_0V_0^{\dagger}-V_1V_1^{\dagger} \bigr)}.
    \]
\end{proof}
\begin{lemma}[Cayley block identity]
\label{lem:cayley-block}
    \begin{equation}
        C(\widetilde L)(\ket 0_A \ket y ) =
        \begin{pmatrix}
            F(L^{\dagger}L)\ket y\\
            -2iL(I+L^{\dagger}L)^{-1}\ket y
        \end{pmatrix}.
        \label{eq:cayley-block}
    \end{equation}
\end{lemma}
\begin{proof}
Write \(\ket w=\ket{w_0}\oplus\ket{w_1}\). The equation \((I+i\widetilde L)\ket w=\ket 0_A \ket y\) gives
\[
    \begin{pmatrix}
        I & iL^\dagger \\
        iL & I
    \end{pmatrix}
    \begin{pmatrix}
        \ket{w_0} \\
        \ket{w_1}
    \end{pmatrix} =
    \begin{pmatrix}
        \ket y\\
        \vec{0}
    \end{pmatrix},
\]
so
\[
    \ket{w_0}+iL^\dagger\ket{w_1}=\ket y,
    \qquad
    \ket{w_1}=-iL\ket{w_0}.
\]
Substitution gives
\[
    \ket{w_0}=(I+L^\dagger L)^{-1}\ket y,
    \qquad
    \ket{w_1}=-iL(I+L^\dagger L)^{-1}\ket y,
\]
and
\begin{align*}
    C(\widetilde L)(\ket 0_A \ket y)
    &=(I-i\widetilde L)\ket w\\
    &=
    \begin{pmatrix}
        \ket{w_0}-iL^\dagger\ket{w_1}\\
        \ket{w_1}-iL\ket{w_0}
    \end{pmatrix}\\
    &=
    \begin{pmatrix}
        (I-L^\dagger L)\ket{w_0}\\
        -2iL\ket{w_0}
    \end{pmatrix}\\
    &=
    \begin{pmatrix}
        F(L^\dagger L)\ket y\\[2mm]
        -2iL(I+L^\dagger L)^{-1}\ket y
    \end{pmatrix}.
\end{align*}
\end{proof}

\begin{proposition}[Local transducer identity]
\label{prop:local-transducer}
We use \(\nu\) as defined in \cref{eq:nu} and \(K_0, K_1\) as defined in \cref{eq:local-kraus}.
\begin{align}
    \nu &= \frac{1-\sqrt{1-\delta}}{1+\sqrt{1-\delta}}
      =\frac{\delta}{(1+\sqrt{1-\delta})^2} \nonumber \\
    K_0 &= (I-\nu \frac{L^\dagger L}{\alpha^2})(I+\nu \frac{L^\dagger L}{\alpha^2})^{-1}, \nonumber \\
    K_1 &= -2i\sqrt{\nu}\,\frac{L}{\alpha}(I+\nu \frac{L^\dagger L}{\alpha^2})^{-1}.\nonumber 
\end{align}
For every \(\ket\psi\in\cH_{\mathrm{S}}\), the catalyst vector
\(\ket{x(\psi)}\) defined by
\cref{eq:caty}--\cref{eq:catx} satisfies
\begin{equation}
    R\bigl(\ket\psi\oplus0\oplus
      \widetilde{O}_L\ket{x(\psi)}\bigr)
    =
    K_0\ket\psi\oplus K_1\ket\psi\oplus\ket{x(\psi)}.
    \label{eq:local-transducer-action}
\end{equation}
\end{proposition}
\begin{proof}
Using \(\widetilde{O}_L^2=I\) and
\(V_{\mathrm{enc}}^\dagger\widetilde{O}_LV_{\mathrm{enc}}=\frac{\widetilde L}{\alpha}\),
\begin{align}
    -iV_{\mathrm{enc}}^\dagger\widetilde{O}_L\ket x 
    &= -iV_{\mathrm{enc}}^\dagger\widetilde{O}_L(I+i\widetilde{O}_L)V_{\mathrm{enc}}\ket z\nonumber\\
    &= (I-i\frac{\widetilde L}{\alpha})\ket z \\
    & = C(\frac{\widetilde L}{\alpha})(\ket 0_A \ket y)\nonumber\\
    &=\begin{pmatrix}
        F(L^\dagger L/\alpha^2)\ket y\\
        -2iL/\alpha(I+L^\dagger L/\alpha^2)^{-1}\ket y
    \end{pmatrix}.
    \label{eq:oracle-cayley}
\end{align}

\begin{enumerate}
    \item Public no-jump:
        \begin{align}
            \sqrt{1-\delta}\ket\psi - i\sqrt\delta V_0^\dagger\widetilde{O}_L \ket x 
            &= \sqrt{1-\delta}\ket\psi + \sqrt\delta F(\frac{L^{\dagger}L}{\alpha^2})\ket y\nonumber\\
            &= \sqrt{1-\delta}\ket\psi \nonumber + \delta F(\frac{L^{\dagger}L}{\alpha^2}) (I+\sqrt{1-\delta} F(\frac{L^{\dagger}L}{\alpha^2}))^{-1}\ket \psi \nonumber \\                    
            &= \bigl[\sqrt{1-\delta}(I+\sqrt{1-\delta}F(\frac{L^{\dagger}L}{\alpha^2}))+\delta F(\frac{L^{\dagger}L}{\alpha^2})\bigr](I+\sqrt{1-\delta}F(\frac{L^{\dagger}L}{\alpha^2}))^{-1}\ket\psi\nonumber\\
            &=(\sqrt{1-\delta}I+F(\frac{L^{\dagger}L}{\alpha^2}))(I+\sqrt{1-\delta}F(\frac{L^{\dagger}L}{\alpha^2}))^{-1}\ket\psi\nonumber\\
            &=(I-\nu \frac{L^{\dagger}L}{\alpha^2})(I+ \nu \frac{L^{\dagger}L}{\alpha^2})^{-1}\ket\psi.
            \label{eq:no-jump-output}
        \end{align}
        For the last equality, use
        \begin{align*}
            \sqrt{1-\delta}I+F(\frac{L^{\dagger}L}{\alpha^2})&=(1+\sqrt{1-\delta})(I-\nu \frac{L^{\dagger}L}{\alpha^2})(I+\frac{L^{\dagger}L}{\alpha^2})^{-1},\\
            I+\sqrt{1-\delta}F(\frac{L^{\dagger}L}{\alpha^2})&=(1+\sqrt{1-\delta})(I+\nu \frac{L^{\dagger}L}{\alpha^2})(I+\frac{L^{\dagger}L}{\alpha^2})^{-1}.
        \end{align*}
    \item Public jump:
        \begin{align}
            -iV_1^\dagger\widetilde{O}_L\ket x &=-2i\frac{L}{\alpha}(I+\frac{L^{\dagger}L}{\alpha^2})^{-1}\ket y \\
            &=-2i\sqrt{\delta}\frac{L}{\alpha}(I+\frac{L^{\dagger}L}{\alpha^2})^{-1}(I+\sqrt{1-\delta} F(\frac{L^{\dagger}L}{\alpha^2}))^{-1}\ket \psi \\
            &=-2i\sqrt{\delta}\frac{L}{\alpha}(I+\frac{L^{\dagger}L}{\alpha^2})^{-1}[I+\sqrt{1-\delta} (I-\frac{L^{\dagger}L}{\alpha^2})(I+\frac{L^{\dagger}L}{\alpha^2})^{-1}]^{-1}\ket \psi \\
            &=-2i\sqrt{\delta}\frac{L}{\alpha}(I+\frac{L^{\dagger}L}{\alpha^2})^{-1}[(1+\sqrt{1-\delta})(I+\frac{\delta}{(1+\sqrt{1-\delta})^2}\frac{L^{\dagger}L}{\alpha^2})(I+\frac{L^{\dagger}L}{\alpha^2})^{-1}]^{-1}\ket \psi \\
            & =-2i\frac{\sqrt{\delta}}{1+\sqrt{1-\delta}}\frac{L}{\alpha}(I+\frac{\delta}{(1+\sqrt{1-\delta})^2}\frac{L^{\dagger}L}{\alpha^2})^{-1}\ket\psi \nonumber\\ 
            &= -2i \sqrt{\nu} \frac{L}{\alpha}(I+\nu \frac{L^{\dagger}L}{\alpha^2})^{-1} \ket \psi.
        \end{align}
    \item Private:
        \cref{eq:catx} and \cref{eq:lift-compression} give
        \[
            V_{\mathrm{enc}}^\dagger\ket x=(I+i\frac{\widetilde{L}}{\alpha})\ket z=\ket 0_A \ket y,
        \]
        so \(V_0V_0^{\dagger}\ket x=V_0\ket y\) and \(V_1V_1^{\dagger}\ket x=0\). 
        \cref{eq:oracle-cayley} also gives
        \[
          iV_0V_0^{\dagger} \widetilde{O}_L \ket x=-V_0F(\frac{L^{\dagger}L}{\alpha^2})\ket y.
        \]
        \(\bigl(I-V_0V_0^{\dagger}-V_1V_1^{\dagger} \bigr) V_{\mathrm{enc}}=0\) and \cref{eq:catx} give
        \[
          i\bigl(I-V_0V_0^{\dagger}-V_1V_1^{\dagger} \bigr) \widetilde{O}_L \ket x=\bigl(I-V_0V_0^{\dagger}-V_1V_1^{\dagger} \bigr)\ket x.
        \]
        Thus the private output is
        \begin{align}
            &\quad \sqrt{\delta} V_0\ket\psi+i(\sqrt{1-\delta}V_0V_0^{\dagger}+\bigl(I-V_0V_0^{\dagger}-V_1V_1^{\dagger} \bigr))\widetilde{O}_L\ket x \\
            &= \sqrt{\delta}V_0 \ket \psi -\sqrt{1-\delta}V_0F(\frac{L^{\dagger}L}{\alpha^2})\ket y \nonumber + \bigl(I-V_0V_0^{\dagger}-V_1V_1^{\dagger} \bigr) \ket x \nonumber \\
            &=V_0\bigl[(I+\sqrt{1-\delta}F(\frac{L^{\dagger}L}{\alpha^2})) \ket y -\sqrt{1-\delta} F(\frac{L^{\dagger}L}{\alpha^2})\ket y\bigr] \nonumber 
            +\bigl(I-V_0V_0^{\dagger}-V_1V_1^{\dagger} \bigr)\ket x\nonumber\\
            &=V_0\ket y+\bigl(I-V_0V_0^{\dagger}-V_1V_1^{\dagger} \bigr) \nonumber\ket x \\
            &=\ket x.
        \end{align}
    \end{enumerate}
\end{proof}

\begin{lemma}[Catalyst size]
\label{lem:catalyst-size}
The catalyst map satisfies
\begin{equation}
    \norm{\Gamma_{\mathrm loc}}^2 \leq 2\delta.
    \label{eq:catalyst-size}
\end{equation}
\end{lemma}
\begin{proof}
Since \(0 \preceq F(\frac{L^{\dagger}L}{\alpha^2}) \preceq I\), \cref{eq:caty}--\cref{eq:catz} imply
\(\norm {y}\leq \sqrt{\delta}\norm{\psi}\) and \(\norm {z}\leq\norm {y}\). Because \(\widetilde{O}_L\) is a Hermitian unitary 
\[
    \norm {x}^2
    = \bra {z}V_{\mathrm{enc}}^\dagger(I-i\widetilde{O}_L)(I+i\widetilde{O}_L)V_{\mathrm{enc}}\ket {z}
    = 2\norm {z}^2
    \leq2\delta\norm{\psi}^2.
\]
\end{proof}

\subsection{Global Transducer}
The local transducer produces two public branches whose induced channel is a first-order approximation to the dissipative evolution. To compose \(r\) local transducers, define the environment register \(\mathrm{E}=\mathrm{E}_0 \mathrm{E}_1\cdots \mathrm{E}_{r-1}\) and a time register on \(\cH_{\mathrm{T}}\). The environment register \(\mathrm{E}\) records the jump history. At step \(j\), we write \(\mathrm{E}_{<j}\) for the environment registers from preceding steps and \(E_{> j}\) for the environment registers still in their initial \(\ket 0\) state that will be set in the future steps. At step \(j\), the environment in the local construction is identified with the accumulated history register \(\mathrm E_{<j}\). The local transducer acts trivially on this register and records the new jump outcome in \(\mathrm E_j\), thereby extending the purification. 

Let
\[
    \cH_{\mathrm{E}} \coloneqq (\mathbb C^2)^{\otimes r},
    \qquad
    \cH_{\mathrm{T}} \coloneqq \operatorname{span}\{\ket j_{\mathrm T}:0\leq j<r\}
    \cong\mathbb C^r.
\]

The public and private spaces are
\[
\underbrace{\mathcal{H}_\mathrm{E}}_{\text{shared history}}
\otimes
\left[
  \underbrace{\cH_{\mathrm{S}}}_{\text{public}}
  \;\oplus\;
  \underbrace{\cH_{\mathrm{T}} \otimes \cH_{\mathrm{C}, \mathrm{loc}}}_{\text{private}}
\right],
\]
\[
    \cH_{\mathrm{P}}
    \coloneqq \cH_{\mathrm{E}}\otimes \cH_{\mathrm{S}} ,
    \qquad
    \cH_{\mathrm{C}}
    \coloneqq \cH_{\mathrm{T}} \otimes \cH_{\mathrm{E}}\otimes\cH_{\mathrm{C}, \mathrm{loc}}.
\]
\(\cH_{\mathrm{C}, \mathrm{loc}}\) as defined in \cref{eq:local-spaces} excludes the common environment. 

Let \(V_{\mathrm{in}}\) attach the all-zero history \(\ket{0^r}_{\mathrm{E}}\):
\[
    V_{\mathrm{in}}:\cH_{\mathrm{S}}\to\cH_{\mathrm{P}}, 
    \qquad 
    V_{\mathrm{in}} := \ket {0^r}_\mathrm{E} \otimes I.
\]

The initial input state is \(\ket \psi\) and each application of the Stinespring isometry from the local transducer produces two new branches. 

For a binary history \(h=(h_0, h_1, \ldots, h_{j-1}) \in \{0, 1\}^j\) where \(1 \leq j \leq r\) define 
\[
    K_h \coloneqq K_{h_{j-1}} K_{h_{j-2}} \cdots K_{h_{0}}, \qquad
    K_{\varnothing} \coloneqq I.
\]
\(K_h\) is the ordered product of Kraus operators associated with the history \(h\).

Define the \(j^{\mathrm{th}}\) purified public state as
\begin{equation}
    \ket{\Psi_j} = \sum_{h\in\{0, 1\}^{j}} \ket {h 0^{r-j}}_{\mathrm{E}} K_{h} \ket\psi_{\mathrm{S}}.
    \label{eq:history-state}
\end{equation}
where \(\ket{\Psi_0} = V_{\mathrm{in}}\ket \psi\). This is the state produced by \(j\) applications of the Stinespring isometry from the local transducer on the input state \(\ket \psi\).

We define the local catalyst at the \(j^{\mathrm{th}}\) step as
\begin{equation}
    \ket {x_j} := (I_{\mathrm{E}} \otimes \Gamma_{\mathrm{loc}})\ket{\Psi_j}.
    \label{eq:stage-catalyst}
\end{equation}

The global catalyst map associates the local catalyst at step \(j\) with the time label \(j\)
\begin{equation}
    \Gamma\ket\psi := \sum_{j=0}^{r-1}\ket j_{\mathrm T} \otimes \ket{x_j}.
    \label{eq:global-catalyst}
\end{equation}

For each \(j=0,\ldots,r-1\) and history prefix \(h\in\{0,1\}^j\) define the embedding isometry \(J_{j, h} : \cH_{\mathrm{P}, \mathrm{loc}} \oplus \cH_{\mathrm{C}, \mathrm{loc}} \rightarrow \cH_{\mathrm{P}} \oplus \cH_{\mathrm{C}}\) by
\begin{align}
    &\quad J_{j, h} \left[(\ket {\psi_0}_{\mathrm{S}^{(0)}} \oplus \ket {\psi_1}_{\mathrm{S}^{(1)}})_{\mathrm{P, loc}} \oplus \ket {\phi}_{\mathrm{C, loc}} \right]  \\
    &= \nonumber 
     \bigl( \ket{h00^{r-j-1}}_{\mathrm{E}} \ket{\psi_0} + \ket {h10^{r-j-1}}_{\mathrm{E}}\ket{\psi_1}\bigr)_{\cH_{\mathrm{P}}} \nonumber
     \oplus \bigl(\ket j_{\mathrm{T}} \ket {h0^{r-j}}_{\mathrm{E}} \ket  \phi \bigr)_{\cH_{\mathrm{C}}}.\nonumber
\end{align}

Let \(R_j\) be the query-free \(j^{\mathrm{th}}\) stage unitary.
\begin{equation}
    R_j \coloneqq \sum_{h\in \{0, 1\}^j} J_{j, h} R J_{j, h}^{\dagger } + (I - \sum_{h\in \{0, 1\}^j} J_{j, h}J_{j, h}^{\dagger}).
    \label{eq:R_j}
\end{equation}

For each history \(h\in\{0, 1\}^{j}\), \(R_j\) acts as \(R\) where public no-jump is identified with \(\ket{h00^{r-j-1}}_\mathrm{E}\otimes\cH_{\mathrm{S}}\), public jump is identified with \(\ket{h10^{r-j-1}}_\mathrm{E}\otimes\cH_{\mathrm{S}}\) and the local private space is identified with \(\ket j_{\mathrm {T}}\otimes\ket{h0^{r-j}}_\mathrm{E}\otimes\cH_{\mathrm{C},\mathrm{loc}}\). It records its public no-jump or jump output in the fresh history register \(\mathrm{E}_j\). 

In chronological evolution, at the start of stage \(j\), every public jump-input sector is empty, and the private component with time label \(j\) is \(\ket j_{\mathrm T}\otimes\ket{x_j}\), where \(\ket{x_j}\) includes the corresponding history strings in the history register.

Let \(O_t\) apply \(\widetilde{O}_L\) on \(\cH_{\mathrm{C}, \mathrm{loc}}\) controlled on the time register being \(t\) and act as the identity on the public and all other private branches. 
The chronological circuit is
\[
    R_{r-1}O_{r-1}R_{r-2}O_{r-2} \cdots R_{0}O_{0}.
\]
Since \(R_j\) acts as the identity on every private sector with time label \(i \neq j\), the oracle block associated with \(O_i\) commutes with \(R_j\) whenever \(i\neq j\). Thus, the oracle blocks can be collected into the single query \(O\)
\begin{equation}
    O \coloneqq \sum_{j=0}^{r-1}\ket j\bra j_{\mathrm{T}}\otimes I_\mathrm{E}\otimes\widetilde{O}_L
    \label{eq:global-oracle}
\end{equation}
and
\begin{align}
    S^\circ &\coloneqq  R_{r-1}R_{r-2} \cdots R_{0}\nonumber\\
    S & \coloneqq S^\circ(I_{\mathrm{P}}\oplus O).
    \label{eq:s-global-transducer}
\end{align}

Induction using Proposition~\ref{prop:local-transducer} gives, on the blank-history input subspace,
\begin{equation}
    S\bigl(V_{\mathrm{in}}\ket\psi \oplus \Gamma\ket\psi\bigr) 
    = V_{r}\ket\psi \oplus \Gamma\ket\psi,
    \label{eq:global-transducer}
\end{equation} 
where
\[
    V_{r}\ket\psi \coloneqq \ket{\Psi_r}.
\] 

\(V_{r}\) is the Stinespring isometry for the \(r\)-step channel:
\begin{equation}
    \Tr_{\mathrm{E}}(V_{r}\rho V_{r}^\dagger)
    =\sum_{h\in\{0,1\}^r}K_h \rho K_h^\dagger
    =\Phi_\delta^r(\rho).
    \label{eq:environment-trace}
\end{equation} 

Since the intermediate public evolution is isometric, orthogonality of the time labels gives
\[
    \|\Gamma\|^2
    \leq r\|\Gamma_{\mathrm{loc}}\|^2
    \leq 2r\delta.
\]

All subsequent reuse identities are restricted to the blank-history input subspace.

\subsection{ Block Structure of the Valid Private-to-Private Map of the Global Transducer}
The standard reuse bound has an unfavorable \(1/\epsilon^2\) scaling with error. To get additive scaling, we aim to use the LCU reuse construction of \cite{CGWZ26}, which can produce cancellation between errors from different reuse lengths. The error in the LCU reuse construction as defined in \cref{eq:reuse-error-poly} has 3 factors. Since \(S_{\mathrm{P} \leftarrow \mathrm{C}}\) and \(\Gamma\) are fixed by our transducer construction, we aim to control the norm of the middle term \(p_{\mathrm{err}}\) evaluated on \(S_{\mathrm{C} \leftarrow \mathrm{C}}\). To get this bound, we analyze the structure of \(S_{\mathrm{C} \leftarrow \mathrm{C}}\) by performing a block decomposition of \(S\) based on the time register.

Write the block decomposition of \(S^\circ\) as defined in \cref{eq:s-global-transducer} as
\begin{equation}
       S^\circ=\begin{pmatrix}
        S^{\circ}_{\mathrm{P} \leftarrow \mathrm{P}} & S^{\circ}_{\mathrm{P} \leftarrow \mathrm{C}} \\
        S^{\circ}_{\mathrm{C}\leftarrow \mathrm{P}} & S^{\circ}_{\mathrm{C} \leftarrow \mathrm{C}}
    \end{pmatrix}. 
\end{equation}

\(S^{\circ}\) sequentially applies the rotations \(R_0, R_1, ..., R_{r-1}\). Let the \(j^{\mathrm{th}}\) portion of the private state be the portion with time label \(j\). Each rotation \(R_j\) acts on part of the public state and the \(j^{\mathrm{th}}\) portion of the private state. The \(j^{\mathrm{th}}\) portion of the private state is unchanged in all subsequent rotations.

By the definition of \(R\) in \cref{eq:R}, the output in the \(j^{\mathrm{th}}\) portion of the private state is determined by the public input to \(R_j\) and the private input. Because we are trying to bound \(S_{\mathrm{C} \leftarrow \mathrm{C}}\) we need only consider the portion of the public state at the \(j^{\mathrm{th}}\) stage determined by the private state. Thus we track paths that originate in the private space, propagate through the public space and return to the private space at a later time. 

In the rotation defined in ~\cref{eq:R}, there is a norm-one public-to-private path from the jump input. In chronological evolution the fresh history register is \(\ket 0\) so the public jump input is \(0\) preventing any amplitude from taking this path. We restrict our analysis to the valid history blocks consistent with chronological evolution which prevent these paths.

To describe the valid subspace and track the valid private-to-public-to-private paths we define for \(0\leq \ell \leq r-1\) and \(h\in\{0,1\}^\ell\), the valid history blocks whose direct sum contains the range of the catalyst
\begin{equation}
    \cV_{\ell,h} 
    \coloneqq \ket \ell_\mathrm{T}\ket{h0^{r-\ell}}_{\mathrm{E}}\otimes\cH_{\mathrm{C}, \mathrm{loc}},
    \qquad
    \cV \coloneqq \bigoplus_{\ell=0}^{r-1}\ \bigoplus_{h\in\{0,1\}^\ell}\cV_{\ell,h}.
    \label{eq:valid-private-space} 
\end{equation}
In \(\cV\), all environment registers \(\mathrm{E}_{\geq\ell}\) are fixed to zero at stage \(\ell\).

The history register \(\mathrm{E}_\ell\) is \(\ket 0\) when stage \(\ell\) begins. For \(k>\ell\), define the history isometries, \(H^{(0)}_{k\ell}, H^{(1)}_{k\ell}\)  on the environment registers \(\mathrm{E}_{\leq \ell}\) by
\begin{align}
    H^{(0)}_{k\ell}\ket {h0^{r-\ell}}_{\mathrm{E}}&:=\ket{h00^{r-\ell-1}}_{\mathrm{E}},\\
    H^{(1)}_{k\ell}\ket {h0^{r-\ell}}_{\mathrm{E}}&:=\ket{h1 0^{r-\ell-1}}_{\mathrm{E}},
\end{align}
which record a no-jump or jump event at time \(\ell\) followed by no-jump events for the remaining history. 

Let
\begin{equation}
    T^{(0)}_{k\ell}:=H^{(0)}_{k\ell}\otimes V_0V_0^{\dagger},
    \qquad
    T^{(1)}_{k\ell}:=H^{(1)}_{k\ell}\otimes V_0V_1^\dagger,
    \label{eq:history-transition-isometries}
\end{equation} 
where \(V_0, V_1\) are defined in \cref{eq:V-local}. \(\mathrm{E}_{>\ell}\) is left in its all 0 state. The \(V_0V_0^{\dagger}\), \(V_0V_1^\dagger\) factors come from the two types of paths from private to public to private.

We consider the block decomposition of \(S_{\mathrm{C} \leftarrow \mathrm{C}}\) based on the value of the time label. For an operator \(A\) on \(\cV\), let
\[
    [A]_{k\ell}:=(\bra k_{\mathrm{T}}\otimes I)A(\ket \ell_{\mathrm{T}}\otimes I)
\] select the block with row time label \(k\) and column time label \(\ell\). 

We split \(S_{\mathrm{C}\leftarrow \mathrm{C}}\) into an unperturbed operator and a correction involving \(\Delta_{\mathrm{R}}, \Delta_{\mathrm{J}}\). The unperturbed operator is annihilated by a polynomial we introduce below. All remaining terms contain at least one of these correction terms.

\(\Delta_{\mathrm{R}}\) contains a diagonal correction to the direct private-to-private transition as well as the paths that depart the private space from the no-jump output. \(\Delta_{\mathrm{J}}\) contains the paths that depart the private space from the jump output. We define these operators through their time-indexed blocks as follows.
\begin{align}
    [{\Delta_{\mathrm{R}}]}_{k\ell} &= \begin{cases}
        0 & \text{if } k < \ell \\ 
        (1-\sqrt{1-\delta})V_0V_0^{\dagger} & \text{if } k = \ell \\
        \delta(\sqrt{1-\delta})^{k-\ell-1}T_{k\ell}^{(0)} & \text{if } k > \ell
    \end{cases}   \label{eq:DeltaR} \\
    {[\Delta_{\mathrm{J}}]}_{k\ell} &= \begin{cases}
        0 & \text{if } k < \ell \\ 
        0 & \text{if } k = \ell \\
        \sqrt{\delta}(\sqrt{1-\delta})^{k-\ell-1}T_{k\ell}^{(1)} & \text{if } k > \ell
    \end{cases} \label{eq:DeltaJ}
\end{align}

Let
\begin{equation}
    \Delta \coloneqq \Delta_{\mathrm{R}}+\Delta_{\mathrm{J}},
    \qquad
    P \coloneqq \bigoplus_{\ell,h}V_1V_1^{\dagger}
    \quad\text{on }\cV.
    \label{eq:Delta-and-P}
\end{equation}

\begin{lemma}[Private Block]
\label{lem:private-block}
The valid private space \(\cV\) is invariant under \(S^{\circ}_{\mathrm{C} \leftarrow \mathrm{C}}\), \(O\), and thus \(S^{\circ}_{\mathrm{C} \leftarrow \mathrm{C}} O\). On the invariant valid private space \(\cV\),
\begin{align}
    S_{\mathrm{C} \leftarrow \mathrm{C}} &= S^{\circ}_{\mathrm{C} \leftarrow \mathrm{C}}O  \\
                      &= i(I-P-\Delta)O \\
                      &= i(I-P)O -i\Delta O.
    \label{eq:S-circ-parts}
\end{align}
\end{lemma}

\begin{proof}
On the diagonal portion, the private-to-private block of \cref{eq:R} is
\[
    i(\sqrt{1-\delta}V_0V_0^{\dagger}+\bigl(I-V_0V_0^{\dagger}-V_1V_1^{\dagger} \bigr))=i\bigl(I-V_1V_1^{\dagger}-(1-\sqrt{1-\delta})V_0V_0^{\dagger}\bigr).
\]

Because we apply the local rotations in chronological order and local rotations only act on the private state with the corresponding time label, all blocks with \(k < \ell\) vanish. 

For \(k>\ell\), there are exactly two private \(\rightarrow\) public \(\rightarrow\) private paths.  
\begin{enumerate}
    \item private \(\rightarrow\) public through no-jump \(\rightarrow\)private
    
    From ~\cref{eq:R}, the no-jump departure contributes
    \(-i\sqrt\delta V_0^\dagger\). The \(k-\ell-1\) public-to-public blocks each contribute \(\sqrt{1-\delta}I\) and the return contributes \(\sqrt\delta V_0\).

    Thus the path is \(-i\delta \sqrt{1-\delta}^{k-\ell-1}T^{(0)}_{k\ell}\).
    \item private \(\rightarrow\) public through jump \(\rightarrow\) private

    From ~\cref{eq:R}, the jump departure contributes \(-iV_1^\dagger\). The rest of the path is the same as above.

    Thus the path is
    \(-i\sqrt\delta(\sqrt{1-\delta})^{k-\ell-1}T^{(1)}_{k\ell}\).
\end{enumerate}

These off-diagonal blocks equal \(-i[\Delta]_{k\ell}\), proving
\cref{eq:S-circ-parts}. The only operation that sets the value of a register \(E_j\) to a nonzero value is the rotation \(R_j\). Thus, these chronological updates preserve \(\cV\).
\end{proof}

\subsection{Factorial Decay}
\label{sec:factorial_decay}
To bound the norm of \(p_{\mathrm{err}}(S_{\mathrm{C}\leftarrow \mathrm{C}})\) for some polynomial \(p_{\mathrm{err}}\) 
to be chosen later, we first bound the norm of \(i\Delta O\). 

A \emph{history-preserving contraction} is an operator \(D\) on \(\cV\) that preserves every time-history block \(\cV_{\ell,h}\) and satisfies \(\norm{D}\leq1\). Equivalently, for every time \(\ell\), it preserves the time and environment registers. Note that a \emph{history-preserving contraction} can be time or history dependent. Products of \(P\), \(O\), and their time-controlled versions have this property.

We first bound the norm of algebraic words with a fixed number of \(\Delta_{\mathrm{R}}\) and \(\Delta_{\mathrm{J}}\) factors.
\begin{lemma}[Algebraic word norm bound]
\label{lem:alg-word}
Let \(\sigma=(\sigma_1,\ldots,\sigma_d)\in\{\mathrm{R}, \mathrm{J}\}^d\), let \(D_0,\ldots,D_d\) be fixed history-preserving contractions, and define
\begin{equation}
    W_\sigma \coloneqq D_d\Delta_{\sigma_d}D_{d-1}\cdots D_1\Delta_{\sigma_1}D_0.
    \label{eq:fixed-typed-word}
\end{equation}
If \(\sigma\) contains \(u\) letters \(\mathrm{R}\) and \(v\) letters \(\mathrm{J}\), then
\begin{equation}
    \norm{W_\sigma}
    \leq
    \delta^{u+v/2}\binom{r+u-1}{u}\binom{r}{v}^{1/2}
    \leq
    \delta^u\binom{r+u}{u}
    \left(\delta^v\binom{r}{v}\right)^{1/2}.
    \label{eq:mixed-word-binomial}
\end{equation}
If \(v>r\), the word and the right-hand side are zero.
\end{lemma}
\begin{proof}
Each \(\Delta_{\mathrm{J}}\) transition strictly increases the time label so there can be no paths when \(v > r\).

We use the operator-valued block Schur test obtained by applying the scalar Schur test \cite[Sec.~3, Eq.~(3.8)]{DK03} to the matrix of block norms. Specifically, for an operator \(A\) with blocks \([A]_{k\ell}\) relative to some orthogonal block decomposition, 
\[
    \norm{A} \leq \sqrt{\left(\max_\ell\sum_k\norm{[A]_{k\ell}}\right)
    \left(\max_k\sum_\ell\norm{[A]_{k\ell}}\right)}.
\] 

We consider the block decomposition of \(W_{\sigma}\) based on the value of the time and environment registers. 

A \(\Delta_{\mathrm{R}}\)-transition may keep its time label and appends only zeros when it advances time. A \(\Delta_{\mathrm{J}}\)-transition strictly advances time and updates the history by \(h0^{k+1}\mapsto h10^{k}\). The intervening \(D_i\)'s change neither the time nor the history. Thus a nonzero block path through \(W_\sigma\) has nondecreasing time labels
\[
    0\leq t_0\leq t_1\leq\cdots\leq t_d\leq r-1.
\]

Since \(1-\sqrt{1-\delta}\leq\delta\), every \(\Delta_{\mathrm{R}}\) block has norm at most \(\delta\), while every \(\Delta_{\mathrm{J}}\) block has norm at most \(\sqrt\delta\). Each block of \(W_\sigma\) is a sum over paths through intermediate time–history blocks. By submultiplicativity, each path contributes an operator of norm at most \(\delta^{u+v/2}\). The triangle inequality therefore bounds each block norm by \(\delta^{u+v/2}\) times the number of paths contributing to that block. It remains to count the total number of paths contributing to each block row and each block column. For an output prefix, the jump departure times are fixed whereas an input prefix does not specify the subsequent jumps. This produces different bounds for the maximum row and column sums of block norms.  

Fix an output time \(k\) and an output history prefix \(h\in\{0,1\}^k\); equivalently, fix the block row associated with \(\cV_{k,h}\). Since only \(\Delta_{\mathrm J}\) transitions add jumps and time never decreases, the departure labels of the \(v\) \(\Delta_{\mathrm{J}}\) transitions are the positions of the last \(v\) ones in its history, so they are fixed. The departure labels of the \(u\) \(\Delta_{\mathrm{R}}\) transitions form a multiset of size \(u\) drawn from \(r\) labels, giving at most \(\binom{r+u-1}{u}\) choices. With the fixed type word, these departure labels determine the complete time sequence and path. 

Hence the maximum block-row sum is at most
\begin{equation}
    \delta^{u+v/2}\binom{r+u-1}{u}.
    \label{eq:mixed-row-sum}
\end{equation}

Fix an input time \(\ell\) and an input history prefix \(h\in\{0,1\}^\ell\); equivalently, fix the block column associated with \(\cV_{\ell,h}\). The arrival labels of the \(v\) \(\Delta_{\mathrm{J}}\) transitions are strictly ordered and have at most \(\binom{r}{v}\) choices. The \(u\) regular-arrival labels form a multiset with at most \(\binom{r+u-1}{u}\) choices. Together with the fixed type word and input block, these determine the output block and the complete path, so the maximum block-column sum is at most
\begin{equation}
    \delta^{u+v/2}\binom{r+u-1}{u}\binom{r}{v}.
    \label{eq:mixed-column-sum}
\end{equation}
The operator-valued block Schur test gives
\[
    \|W_\sigma\|\leq\sqrt{\delta^{u+v/2}\binom{r+u-1}{u} \delta^{u+v/2}\binom{r+u-1}{u}\binom{r}{v}} = \delta^{u+v/2}\binom{r+u-1}{u}\binom{r}{v}^{1/2}.
\]
\end{proof}

We now define a polynomial \(p_{\mathrm{annih}}(\cdot)\) and bound the norm of \(p_{\mathrm{annih}}(\cdot)\) evaluated on \(S_{\mathrm{C}\leftarrow \mathrm{C}}\). We use this polynomial in the construction of \(p_{\mathrm{err}}(\cdot)\).

Define
\begin{equation}
    \qquad p_{\mathrm{annih}}(A):=\frac{A^2+A^4}{2}.
    \label{eq:k-polynomial}
\end{equation}

\begin{lemma}[Error Cancellation Polynomial]
\label{lem:k-annihilation}
On \(\cV\),
\begin{equation}
    (i(I-P)O)^2=-I+P+OPO,
    \qquad
    p_{\mathrm{annih}}(i(I-P)O)=0.
    \label{eq:k-annihilates-part1}
\end{equation}
Additionally, \(\norm{ p_{\mathrm{annih}}(B)}\leq1\) for every contraction \(B\).
\end{lemma}
\begin{proof}
On every time-history block, \(P=V_1V_1^{\dagger}\), and the off-diagonal form of \(O\) gives
\[
    V_1V_1^{\dagger} O V_1V_1^{\dagger} = V_1(V_1^\dagger OV_1)V_1^\dagger=0.
\]
Thus \(POP=0\) on \(\cV\).  Using \(O^2=I\),
\[
    -(I-P)O(I-P)O = -I + P + OPO.
\]
The operator \(Q:=P+OPO\) is a projection because it is the sum of two orthogonal projections. Hence \((i(I-P)O)^2=-I+Q\), \((i(I-P)O)^4=I-Q\), and
\[
    p_{\mathrm{annih}}(i(I-P)O)=\frac{1}{2}((i(I-P)O)^2 + (i(I-P)O)^4)=0.
\]
Finally,
\(\norm{p_{\mathrm{annih}}(B)} \leq \frac{\norm{B^2}+\norm{B^4}}{2} \leq1\).
\end{proof}

Because \(p_{\mathrm{annih}}\) annihilates \((i(I-P)O)\), every surviving term in the noncommutative expansion of \( p_{\mathrm{annih}}(S_{\mathrm{C} \leftarrow \mathrm{C}})^q\) contains at least \(q\) transition factors of the form \(-i(\Delta)O\). After expanding \(\Delta = \Delta_{\mathrm{R}} + \Delta_{\mathrm{J}}\), we can bound these using \cref{lem:alg-word}.

\begin{proposition}[Factorial private-block decay]
\label{prop:factorial-decay}

There are universal constants \(C,C^*>0\) such that, whenever
\(q,r\in\mathbb N\) and \(\delta=\tau/r\) satisfy
\begin{equation}
    q \geq C(1+\tau),
    \qquad
    r \geq 4q,
    \label{eq:decay-conditions}
\end{equation}
the following bound holds on the valid private space:
\begin{equation}
    \norm{p_{\mathrm{annih}}(S_{C\leftarrow C})^q\big|_{\cV}}
    \leq
    \left(C^*\sqrt{\frac{\tau}{q}}\right)^q.
    \label{eq:factorial-decay}
\end{equation}
\end{proposition}
\begin{proof}
By \cref{eq:S-circ-parts} and
\(p_{\mathrm{annih}}\bigl(i(I-P)O\bigr)=0\), on \(\cV\) we have
\begin{equation}
    p_{\mathrm{annih}}(S_{\mathrm{C} \leftarrow \mathrm{C}})
    =\frac{1}{2}\sum_{j\in\{2,4\}}
      \left[
        \bigl(i(I-P)O-i\Delta O\bigr)^j
        -\bigl(i(I-P)O\bigr)^j
      \right].
    \label{eq:K-noncommutative-expansion}
\end{equation}
Expanding noncommutatively with \(-i\Delta O=-i\Delta_{\mathrm{R}}O-i\Delta_{\mathrm{J}}O\), every surviving term contains at least one transition factor. The remaining factors are history-preserving contractions, so each term has the form covered by Lemma~\ref{lem:alg-word}.
\begin{equation}
    \frac{1}{2} \left(
        \sum_{j=1}^{2}\binom{2}{j}2^j+
        \sum_{j=1}^{4}\binom{4}{j}2^j
    \right)
    =\frac{1}{2}\bigl((3^2-1)+(3^4-1)\bigr)=44.
    \label{eq:typed-coefficient-weight}
\end{equation}
Therefore the total absolute coefficient weight in \(p_{\mathrm{annih}}(S_{\mathrm{C} \leftarrow \mathrm{C}})^q\) is at most \(44^q\), and every typed word contains
\begin{equation}
    q\leq d:=u+v\leq4q
    \label{eq:transition-count-range}
\end{equation}
transition factors.

Fix such a word and write \(x:=u/d\) and \(y:=v/d\). Since \(x,y\geq0\) and \(x+y=1\),
\[
    (2e)^x e^{y/2}x^{-x}y^{-y/2} \leq4e.
\]
The inequalities \(d\geq q\geq C\tau\), \(\tau/d\leq1\), \(u \leq d \leq 4q \leq r\), Lemma~\ref{lem:alg-word}, and the standard binomial estimates imply
\[
    \norm{W_\sigma}
    \leq
    \left[
        4e\left(\frac{\tau}{d}\right)^{x+y/2}
    \right]^d
    \leq
    \left(4e\sqrt{\frac{\tau}{d}}\right)^d.
\] 

Choose \(C\geq 16e^2\) in \cref{eq:decay-conditions}. Then \(4e\sqrt{\tau/q}\leq1\), and for \(d\geq q\),
\[
    \left(4e\sqrt{\frac\tau d}\right)^d
    \leq
    \left(4e\sqrt{\frac\tau q}\right)^q.
\]
Absorbing the coefficient weight \(44^q\) gives the claim. For example \(C^*=176e\) suffices.
\end{proof}

\subsection{LCU on reuse blocks}
We now convert these bounds into an implementable algorithm. Multiplication by the moving average polynomial defined in \cref{eq:poly-g} expresses the error polynomial in the reuse basis. We then use one round of oblivious amplitude amplification to implement the simulation channel.

Following the moving-average change of basis in \cite[Sec.~7.1, Eq.~(48)]{CGWZ26}, write
\[
    p_{\mathrm{err}}(A)\coloneqq p_{\mathrm{annih}}(A)^q g_{6q}(A) = \sum_{j=0}^{10q-1} a_jA^j,
    \qquad a_{10q}:=0,
\]
and define
\[
    \lambda_i:=i(a_{i-1}-a_i),
    \qquad 1\leq i\leq10q.
\]
Then
\[
    p_{\mathrm{err}}(A) = \sum_{i=1}^{10q}\lambda_ig_i(A),
    \qquad
    \sum_{i=1}^{10q}\lambda_i =1.
\]

To evaluate the LCU normalization, write
\[
    p_{\mathrm{annih}}(A)^q=\sum_{j=0}^{4q}b_jA^j.
\]
These coefficients are nonnegative, have total mass one, and have mean degree
\[
    \mu \coloneqq \sum_{j=0}^{4q} jb_j=3q.
\]

Set \(b_j=0\) outside \(0\leq j\leq 4q\).  Since multiplication by \(g_{6q}\) averages \(6q \) consecutive coefficients,
\begin{equation}
    a_{i-1} - a_i = \frac{b_{i-6q}-b_i}{6q},
    \qquad
    \lambda_i = \frac{i}{6q}(b_{i-6q}-b_i),
    \label{eq:lambda-coefficients}
\end{equation}
and
\begin{align}
    \sum_{i=1}^{10q}|\lambda_i|
    &=\frac{1}{6q}\sum_{j=1}^{4q}jb_j
    +\frac{1}{6q}\sum_{j=6q}^{10q}jb_{j-6q}\nonumber\\
    &=\frac{\mu}{6q}+\frac{6q+\mu}{6q}
    =1+\frac{2\mu}{6q}=2.
 \label{eq:lambda-l1-norm}
\end{align}
Let
\[
    \widetilde V_{r}
    :=\left(\sum_{i=1}^{10q}\lambda_iP_i\right)V_{\mathrm{in}}.
\]
Combining Lemma~\ref{lemma:lcu-error-reduction} with the private-block and catalyst bounds gives a block encoding of \(\widetilde V_{r}\) with normalization \(2\), satisfying
\begin{equation}
    \norm{V_{r}-\widetilde V_{r}}
    \leq
    \sqrt{2\tau}
    \left(C^*\sqrt{\frac{\tau}{q}}\right)^q.
  \label{eq:removal-error}
\end{equation}

Let \(\eta = \norm{V_{r} - \widetilde{V}_{r}}\). One round of robust oblivious amplitude amplification \cite[Theorem~28]{GSLW19}, using the source subspace \(\ran V_{\mathrm{in}}\) and the public output subspace, with all block-encoding auxiliary registers fixed to zero in both, yields a success block \(B\) satisfying
\begin{equation}
    \norm{B-V_{r}}\leq3\eta.
    \label{eq:amplified-success-error}
\end{equation}

Discarding the purification and work registers and retaining the system output in both the success and failure branches gives
\[
    \Phi_{\mathrm{amp}}(\rho)
    =
    \Tr_\mathrm{E}(B\rho B^\dagger)
    +
    \cF(\rho).
\] where \(\cF\) is the CP map associated with the failure branch. \(\norm{\cF}_{\diamond} = \norm{I-B^{\dagger}B}\).

Let \(\Phi_r(\rho):=\Tr_\mathrm{E}(V_{r}\rho V_{r}^\dagger)\).
\begin{align}
    \norm{\Phi_{\mathrm{amp}}-\Phi_r}_\diamond
    &\leq
    (\norm {B}+\norm {V_{r}}) \norm{B-V_{r}}
    +\norm{I-B^\dagger B} \\
    &\leq4\norm{B-V_{r}}
    \leq12\eta.
    \label{eq:amplified-channel-error}
\end{align}
Each LCU invocation uses at most \(10q\) calls to the one-query transducer or its inverse. The amplification circuit uses at most \(30q\) transducer calls.

We use the following lemma to choose the degree needed for the error bound.
\begin{lemma}[Factorial Degree Requirement]
\label{lem:factorial-inversion}
For every fixed \(B\geq1\), there is a constant \(C_B\) such that every
\[
    q\geq C_B\bigl(c(\tau,\epsilon)\bigr)
\]
satisfies
\begin{equation}
    B\sqrt{\tau}
    \left(B\sqrt{\frac{\tau}{q}}\right)^q
    \leq\epsilon
    \qquad(\tau\geq0, 0<\epsilon\leq1/2).
    \label{eq:factorial-inversion}
\end{equation}
\end{lemma}
\begin{proof}
The claim is immediate for \(\tau=0\). For \(\tau>0\), let \(p \coloneqq \frac{q}{2}\).
Then
\[
    B\sqrt{\tau}
    \left(B\sqrt{\frac{\tau}{q}}\right)^q  
    =
    B\sqrt{\tau}
    \left(\frac{B^2\tau}{2p}\right)^p.
\]
The required choice of \(p\), and hence of \(q\), follows from the standard factorial-tail inversion; see
\cite[Sec.~7.3, proof of Thm.~9]
{CGWZ26}.

\end{proof}

\subsection{Proof of Theorem~\ref{thm:single-jump}}

Let \(q\) be an integer such that
\begin{equation}
    q = \lceil (C\bigl(c(\tau,\epsilon)\bigr))\rceil,
    \label{eq:q-choice}
\end{equation}
where \(C\) is a sufficiently large universal constant. Then choose
\begin{equation}
    r:=\left\lceil\max\left\{
        4q,\,2\tau,\,\frac{4C_{\mathrm{loc}}\tau^2}{\epsilon}
        \right\}\right\rceil,
    \qquad
    \delta=\frac\tau r.
    \label{eq:J-choice}
\end{equation}
Corollary~\ref{cor:discretization} makes the channel-discretization error at most \(\frac{\epsilon}{4}\). Lemma~\ref{lem:factorial-inversion}, applied with a sufficiently large fixed constant, and \cref{eq:removal-error} give
\begin{equation}
    \eta=\norm{\widetilde{V}_r-V_r}\leq\frac\epsilon{48}.
    \label{eq:eta-choice}
\end{equation}
The channel estimate \cref{eq:amplified-channel-error} gives
\begin{align}
    \norm{\Phi_{\mathrm amp}-\exp(\tau \frac{1}{\alpha^2}\cD_{L})}_\diamond
    &\leq12\eta + \norm{\Phi_\delta^r - \exp(t\cD_{L})}_\diamond \nonumber\\
    &\leq\frac\epsilon4+\frac\epsilon 4 \leq\epsilon.
    \label{eq:final-channel-error}
\end{align}

Each global-transducer application uses one time-controlled \(\widetilde{O}_L\) query. The reuse construction, LCU step, and amplification use at most \(30q=\cO(q)\) such applications. A call to \(\widetilde{O}_L\) uses at most one controlled call each to \(O_L\) and \(O_L^\dagger\), so this is at most \(60q=\cO(q)\) queries to the original oracle or its inverse, proving
\cref{eq:query-complexity}. A direct implementation of the query-free work performs \(\cO(qr)\) local rotations and history swaps, apart from coefficient-preparation and rotation-synthesis overheads. 
\section{The main theorem}
\label{sec:main-thm}
For \(c(\cdot)\) defined in \cref{eq:c-optimal-function}, the precise statement of our main theorem is given as follows.
\begin{theorem}[Lindblad Simulation with Optimal Query Complexity]
\label{thm:main-thm}
Let \(t>0\) and \(0<\epsilon\leq1/2\). Given controlled access to \(O_H, O_{L_1}, O_{L_2}, ... O_{L_m}\) and their adjoints, there is a quantum channel \(\Phi_{\mathrm{sim}}\) such that
\begin{equation}
    \norm{\Phi_{\mathrm{sim}}-
    \exp\bigl(\tau \frac{\cL}{\norm{\cL}_{\mathrm{be}}}\bigr)}_\diamond 
    \leq\epsilon.
\end{equation}

The channel uses
\begin{equation}
    \cO\left(c(\tau,\epsilon)\right)
    \label{eq:main-query-complexity}
\end{equation}
queries to each of \(O_H, O_{L_1}, O_{L_2}, ..., O_{L_m}\) and their adjoints where \(c(\cdot)\) is defined in \cref{eq:c-optimal-function}.

The implementation of the channel uses \(\cO\left(c(\tau, \epsilon) \left[C_{\mathrm{anc}} + m + \log(e + \frac{\tau+1}{\epsilon})\log(e +\frac{(m+1)(\tau+1)}{\epsilon})\right]\right)\) additional arbitrary one- and two-qubit gates where \(C_{\mathrm{anc}}\) includes the cost of membership tests for \(\ran B_{\mathrm{in}}, \ran B_{\mathrm{out}}\) of the block encodings of the Hamiltonian and jump operators as well as the associated fixed embedding operations. For the standard block encodings used here, \(C_{\mathrm{anc}}\) scales linearly with the number of ancilla qubits used in the block encodings.
\end{theorem}

We now show how to extend the single-jump algorithm to the general case. We first show how to simulate an arbitrary number of jump operators. Then we incorporate a nonzero Hamiltonian by interleaving our transducer with the Hamiltonian Cayley transducer in \cite{CGWZ26}. 

We define the normalized time steps for the Hamiltonian and dissipative portions,
\[  
    \Lambda_{\mathrm{D}} \coloneqq \sum_{i=1}^{m} \alpha_i^2, 
    \qquad
    \delta_H \coloneqq \alpha_H \frac{t}{r}, 
    \qquad 
    \delta_{\mathrm{D}} \coloneqq \Lambda_{\mathrm{D}} \frac{t}{r}.
    \qquad
    \delta_{\cL} \coloneqq \delta_H + \delta_{\mathrm{D}} = \norm{\cL}_{\mathrm{be}}  \frac{t}{r} = \frac{\tau}{r}.
\]

\subsection{Multiple Jump Operators}
To handle multiple jump operators, we slightly modify the construction of \(R\) defined in  \cref{eq:R} in the local transducer to output the jump label register in the public jump branch. We also modify the global transducer such that each \(\mathrm{E}_j\) takes values from \(0, 1, ..., m\) where \(0\) denotes no jump and \(1, 2, ..., m\) identify the jumps. 

We use the stacked jump operator \(\widehat{L}\) defined in \cref{eq:normalized-stacked-jump} in place of \(L\) and apply the local construction with \(L = \widehat{L}, \alpha = 1\).

Using the construction of \(O_{\widehat{L}}\) defined directly after \cref{eq:normalized-stacked-jump} and fixed, known, query-free isometries
\[
    B_{\mathrm{in}}:\cH_{\mathrm{S}}\to\cH_{\mathrm{enc}},
    \qquad
    B_{\mathrm{out}}:\cH_\mathrm{J}\otimes\cH_{\mathrm{S}}\to\cH_{\mathrm{enc}},
    \qquad
    B_{\mathrm{out}}^\dagger O_{\widehat{L}} B_{\mathrm{in}}=\widehat{L},
\]
define
\begin{equation}
    \widetilde{O}_{\widehat{L}}
    :=\begin{pmatrix}
        0 & O_{\widehat{L}}^\dagger \\
        O_{\widehat{L}} & 0\end{pmatrix},
    \qquad
    V_0:=\ket0_A\otimes B_{\mathrm in},
    \qquad
    V_1:=\ket1_A\otimes B_{\mathrm out}.
    \label{eq:combined-v}
\end{equation}

Then \(\widetilde{O}_{\widehat{L}}\) is a Hermitian involution,
\[
    V_{1}^\dagger\widetilde{O}_{\widehat{L}}V_{0}=\widehat{L},
    \qquad
    V_{b}^\dagger\widetilde{O}_{\widehat{L}}V_{b}=0 
    \quad b \in\{0, 1\}.
\]

We use the same block formula for the rotation \(R\) in \cref{eq:R}, with public space \(\cH_{\mathrm{S}} \oplus (\cH_{\mathrm{J}} \otimes \cH_{\mathrm{S}})\) and private space \(\cH_{\mathrm{A}} \otimes \cH_{\mathrm{enc}}\). We omit the shared environment register. By \cref{eq:stacked-first-order-error} the induced channel is a first-order approximation to dissipative evolution. 

Each history register now has alphabet \(\{0,1,\ldots,m\}\); a nonzero symbol records a jump index and its position records jump time. The jump type register, \(\cH_{\mathrm{J}}\), remains inside each operator block in the Schur estimate. 
\subsection{Non-zero Hamiltonian}
For the non-zero Hamiltonian case, we interleave Hamiltonian and dissipative simulation using the algorithm in \cite{CGWZ26} and our dissipative simulation algorithm. 

To combine access to the Hamiltonian and \(\widehat{L}\) we define the local combined SELECT primitive
\begin{equation}
    O_{\mathrm{comp}}^{\mathrm{loc}}
    \coloneqq \widetilde{O}_H \oplus \widetilde{O}_{\widehat{L}}.
    \label{eq:static-extended-select}
\end{equation} for \(\widetilde{O}_H\) defined below in  ~\cref{eq:tilde-H}. 

A direct implementation uses at most one controlled call to each input oracle and its adjoint so each SELECT invocation costs at most two queries to each input oracle.

We briefly restate the algorithm in \cite{CGWZ26} for convenience and then analyze the combined algorithm. 
\subsubsection{Cayley Transducer}
The Cayley transducer requires a Hermitian block-encoding oracle. For a general block-encoding oracle \(O_H\), we use the Hermitian embedding \(\widetilde{O}_H\) and isometry \(V_H\) to obtain a Hermitian oracle that block encodes the same matrix. Let
\begin{equation}
    \widetilde{O}_H \coloneqq
    \begin{pmatrix}
        0 & O_H^\dagger\\
        O_H & 0
    \end{pmatrix},
    \qquad
    V_H \coloneqq \frac1{\sqrt2}
    \begin{pmatrix}
        B_{H,{\mathrm in}}\\
        B_{H,{\mathrm out}}
    \end{pmatrix}.
    \label{eq:tilde-H}
\end{equation}
Then \(V_H\) is an isometry, and
\begin{equation}
    \widetilde{O}_H^\dagger=\widetilde{O}_H,
    \qquad
    \widetilde{O}_H^2=I,
    \qquad
    V_H^\dagger\widetilde{O}_HV_H
    =\frac{H^\dagger+H}{2\alpha_H}=\frac{H}{\alpha_H}.
    \label{eq:H-oracle-assumption}
\end{equation}

The Hamiltonian Cayley step is
\begin{equation}
    U_{H, \mathrm{cay}} := (I-i\delta_H \frac{H}{2\alpha_H})(I+i\delta_H \frac{H}{2\alpha_H})^{-1}.
    \label{eq:H-Cayley}
\end{equation}
Let
\begin{equation}
    c_H:=F(\frac{\delta_H}{2})=\frac{1-\delta_H/2}{1+\delta_H/2},
    \qquad
    s_H:=\sqrt{1-c_H^2},
    \qquad
    P_H:=V_HV_H^\dagger.
    \label{eq:H-constants-proj}
\end{equation}
The query-free rotation
\begin{equation}
    R_H:=
    \begin{pmatrix}
        c_H I_{\mathrm{S}} & -is_HV_H^\dagger\\
        s_HV_H & i(c_HP_H+I-P_H)
    \end{pmatrix}
    \label{eq:H-R}
\end{equation}
is used by the Cayley transducer of \cite[Lemma~2]{CGWZ26}. The corresponding Hamiltonian Cayley transducer is
\[
    S_H \coloneqq R_H \bigl(I \oplus \widetilde O_H\bigr).
\]
We suppress the environment register in the following definition as both operators act as the identity on it. Define the auxiliary vector \(z_H\) and catalyst vector \(x_H\) by
\[
    \ket{z_H}
    \coloneqq
     \sqrt{\delta_H/2}(I_{\mathrm{S}}+i\frac{\delta_H}{2}\frac{H}{\alpha_H})^{-1}\ket\psi,
    \qquad
    \ket{x_H}
    \coloneqq
    (I +i\widetilde{O}_H)V_H \ket {z_H}.
\]

Define the catalyst map as follows.
\begin{equation}
    \Gamma_{H, \mathrm{loc}}\ket \psi \coloneqq \ket {x_H(\psi)}.
\end{equation}
It follows that \(\norm{\Gamma_{H, \mathrm{loc}}}^2 \leq \delta_H\). 

Note that for our composition we treat \(\ket \psi\) as the purified state. 

\subsubsection{Error Bound From Composition}
Define
\[
    \Phi_{\delta_H}(\rho)
    :=U_{H, \mathrm{cay}}\rho(U_{H, \mathrm{cay}})^\dagger,
    \qquad
    \Phi_{\delta_H,\delta_{\mathrm{D}}}
    :=\Phi_{\delta_{\mathrm{D}}}\circ\Phi_{\delta_H}.
\]
The standard local-error estimates, Lie--Trotter bounds, and telescoping give the following bounds for some universal constant \(C_{\mathrm{comp}}\) whenever \(\delta_{\cL}\leq1/2\).
\begin{align}
    \norm{
        \Phi_{\delta_H,\delta_{\mathrm{D}}}
        -e^{\delta_{\cL} \frac{\cL}{\norm{\cL}_{\mathrm{be}}}}
    }_\diamond
    &\leq C_{\mathrm{comp}}\delta_{\cL}^2 \nonumber \\
    \norm{
        \Phi_{\delta_H,\delta_{\mathrm{D}}}^{\,r}
        -e^{t \cL}
    }_\diamond 
    &\leq
    C_{\mathrm{comp}}\frac{\tau^2}{r}.
    \label{eq:combined-global-error}
\end{align}

\subsubsection{Composition of Hamiltonian and Dissipative Transducers}

We compose the Hamiltonian and dissipative transducers by combining their private spaces using a direct sum and interleaving the local Hamiltonian and local dissipative rotations. 

Let \(\cH_{\mathrm{C},H}\) and \(\cH_{\mathrm{C},\mathrm{D}}\) denote the full padded private spaces of the Hamiltonian and stacked dissipative local transducers omitting the shared environment, and let \(\Gamma_{H,{\mathrm{loc}}}, \Gamma_{D,{\mathrm{loc}}}\) be their local catalyst maps. 

For this composition the public and private spaces are
\begin{align}
    \cH_{\mathrm{E},{\mathrm{comp}}} & \coloneqq (\mathbb C^{m+1})^{\otimes r}  \nonumber \\
    \cH_{\mathrm{P},{\mathrm{comp}}} & \coloneqq \cH_{\mathrm{E},{\mathrm{comp}}}\otimes \cH_{\mathrm{S}} \nonumber\\
    \cH_{\mathrm{C}, {\mathrm{comp}}} & \coloneqq \cH_{\mathrm{T}} \otimes \cH_{\mathrm{E},{\mathrm{comp}}} \otimes (\cH_{\mathrm{C}, H} \oplus \cH_{\mathrm{C}, \mathrm{D}}) \nonumber
\end{align}
\[
\underbrace{\cH_{\mathrm{E},{\mathrm{comp}}}}_{\text{shared environment}} 
\otimes
\left[
\underbrace{\begin{matrix}
    \cH_{\mathrm{S}} 
\end{matrix}}_{\text{public}}
\;\oplus\;
\underbrace{\begin{matrix}
    \cH_{\mathrm T} \otimes (\cH_{\mathrm{C}, H} \oplus \cH_{\mathrm{\mathrm{C}, \mathrm{D}}}) 
\end{matrix}}_{\text{private}}
\right].
\] 

At dissipative stage \(j\), we identify the jump-output space \(\cH_{\mathrm{J}} \otimes \cH_{\mathrm{S}}\) with \(\operatorname{span}\{\ket1,\ldots,\ket m\}_{\mathrm{E}_j} \otimes \cH_{\mathrm S}\), reserving \(\ket0_{\mathrm{E}_j}\otimes\cH_{\mathrm S}\) for the no-jump output.

Lift every local transducer to use the common public space and common private space.

Let \(S_{H,j}\) and \(S_{\mathrm{D},j}\) denote the lifts of the Hamiltonian and dissipative local transducers to time step \(j\). For each history prefix in \(\mathrm{E}_{<j}\), the Hamiltonian rotation couples the public system to the corresponding Hamiltonian private sector at time \(j\), with \(\mathrm{E}_j,\ldots,\mathrm{E}_{r-1}\) fixed to \(\ket0\). It leaves the history registers unchanged and acts as the identity outside these subspaces. The dissipative rotation uses the compatible history blocks from the single-jump construction, with jump labels extended to \(1,\ldots,m\). Define the transducer for one step by
\[
    S_{j} \coloneqq S_{\mathrm{D},j}S_{H,j}.
\]
Its induced public channel is
\[
    \Phi_{\delta_H,\delta_\mathrm{D}}=\Phi_{\delta_\mathrm{D}}\circ\Phi_{\delta_H}.
\]
The full chronological composition is
\[
    S_{\mathrm{comp}}
    \coloneqq S_{r-1}\cdots S_{1}S_{0}.
\]
On the global private space, define the combined SELECT oracle
\begin{equation}
    O_{\mathrm{comp}} \coloneqq
       \sum_{j=0}^{r-1}
        \ket j\bra j_\mathrm{T} \otimes I_{\mathrm{E},{\mathrm{comp}}} \otimes
        \bigl(\widetilde{O}_H \oplus \widetilde{O}_{\widehat{L}}\bigr).
  \label{eq:extended-select-oracle}
\end{equation}

An oracle acting on a later time-and-stage sector commutes through all earlier query-free rotations, so the chronological product factors as
\begin{equation}
    S_{\mathrm{comp}}
    =S_{\mathrm{comp}}^\circ
    \left(
        I\oplus O_{\mathrm{comp}}
    \right).
    \label{eq:extended-one-query-factorization}
\end{equation}

At time step \(j\), apply the Hamiltonian stage and then the dissipative stage. Let \(\ket{\Psi_j}\) be the purified public state just before the Hamiltonian stage, and let \(\ket{\Psi_{j+1/2}}\) be the state after that stage. 

With a time-and-stage private label, define
\begin{equation}
    \Gamma_{\mathrm{comp}}\ket\psi
    \coloneqq 
    \sum_{j=0}^{r-1} 
    \left(
    \ket{j}_\mathrm{T} (I_\mathrm{E} \otimes \Gamma_{H,{\mathrm{loc}}})\ket{\Psi_j}
    \oplus \ket{j}_\mathrm{T} (I_\mathrm{E} \otimes \Gamma_{\mathrm{D},{\mathrm{loc}}})\ket{\Psi_{j+1/2}}
    \right).
    \label{eq:global-cat-map}
\end{equation}

The local catalyst bounds
\[
    \norm{\Gamma_{H,{\mathrm{loc}}}}^2\leq \delta_H,
    \qquad
    \norm{\Gamma_{\mathrm{D},{\mathrm{loc}}}}^2\leq2\delta_{\mathrm{D}},
\] 
together with the orthogonality of the private labels, give
\begin{equation}
    \norm{\Gamma_{\mathrm{comp}}}^2
    \leq r\delta_H+2r\delta_{\mathrm{D}}
    \leq 2\tau.
    \label{eq:extended-catalyst-bound}
\end{equation}

The local transducer identities give
\[
    S_{\mathrm{comp}}\bigl(V_{\mathrm in}\ket\psi\oplus
    \Gamma_{\mathrm{comp}}\ket\psi\bigr)
    =V_{\mathrm{comp}}\ket\psi\oplus\Gamma_{\mathrm{comp}}\ket\psi,
\]
where \(V_{\mathrm in}\) now attaches the all-zero \((m+1)\)-ary history and \(V_{\mathrm{comp}}\) is a Stinespring isometry of \(\Phi_{\delta_H,\delta_{\mathrm{D}}}^{\,r}\).

Thus the combined transducer is a one-SELECT-query transducer for a Stinespring isometry of \(\Phi_{\delta_H,\delta_{\mathrm{D}}}^{\,r}\).

The following private-block identities are restricted to the valid private space.

\paragraph{Private Block and Error}
In the following section we treat \(S, S^{\circ}\) as their composed versions. 
As in the single-jump case, we need to bound the norm of a polynomial \(p_{\mathrm{err}}\) evaluated on \(S_{\mathrm{C}\leftarrow \mathrm{C}}\). The main argument follows a similar structure. We decompose this operator into time-indexed blocks determined by the local transducers, separate an unperturbed part annihilated by \(p_{\mathrm{annih}}\), and bound the remaining terms.

At a fixed physical time, write the query-free local stages as
\begin{equation}
    S^\circ_{\sigma, j}=\begin{pmatrix}
        S^{\circ}_{\mathrm{P} \leftarrow \mathrm{P}, \sigma, j} &  S^{\circ}_{\mathrm{P} \leftarrow \mathrm{C}, \sigma, j} \\
        S^{\circ}_{\mathrm{C} \leftarrow \mathrm{P}, \sigma, j} & S^{\circ}_{\mathrm{C} \leftarrow \mathrm{C}, \sigma, j}
    \end{pmatrix}, 
    \qquad 
    \sigma\in\{H,\mathrm{D}\}.
\end{equation}
With the ordering \(\cH_{{\mathrm{P}},{\mathrm{comp}}}\oplus\cH_{\mathrm{C},H} \oplus\cH_{\mathrm{C},\mathrm{D}}\), their lifts with time and environment labels suppressed are
\[
    \widetilde S_{H,j}^\circ=
    \begin{pmatrix}
        S^{\circ}_{\mathrm{P} \leftarrow \mathrm{P}, H, j} & S^{\circ}_{\mathrm{P} \leftarrow \mathrm{C}, H, j} & 0\\
        S^{\circ}_{\mathrm{C} \leftarrow \mathrm{P}, H, j} & S^{\circ}_{\mathrm{C} \leftarrow \mathrm{C}, H, j} & 0\\
        0 & 0 & I_{\cH_{\mathrm{C}, \mathrm{D}}}
    \end{pmatrix},\qquad
    \widetilde S_{D,j}^\circ=
    \begin{pmatrix}
        S^{\circ}_{\mathrm{P} \leftarrow \mathrm{P}, \mathrm{D}, j} & 0 & S^{\circ}_{\mathrm{P} \leftarrow \mathrm{C}, \mathrm{D}, j}\\
        0 & I_{\cH_{\mathrm{C}, H}} & 0\\
        S^{\circ}_{\mathrm{C} \leftarrow \mathrm{P}, \mathrm{D}, j} & 0 & S^{\circ}_{\mathrm{C} \leftarrow \mathrm{C}, \mathrm{D}, j}
    \end{pmatrix}.
\]
Direct multiplication gives the grouped diagonal private block
\begin{equation}
    [S^{\circ}_{\mathrm{C} \leftarrow \mathrm{C}}]_{j,j}
    =
    \begin{pmatrix}
        S^{\circ}_{\mathrm{C} \leftarrow \mathrm{C}, H, j} & 0\\
        S^{\circ}_{\mathrm{C} \leftarrow \mathrm{P}, \mathrm{D}, j} S^{\circ}_{\mathrm{P} \leftarrow \mathrm{C}, H, j} & S^{\circ}_{\mathrm{C} \leftarrow \mathrm{C}, \mathrm{D}, j}
    \end{pmatrix}.
    \label{eq:grouped-diagonal-block}
\end{equation}
The lower-left entry is the only new same-time private path.

For the norm estimates, we group full histories according to the times at which jumps occur. We label these groups by binary strings \(h\), where \(h_k=0\) denotes no jump and \(h_k=1\) denotes a jump of any type. The individual jump labels remain within each operator block, so the block-norm estimates account for all jump types collectively. This allows us to reuse the preceding row and column counting without introducing a dependence on \(m\). For \(0\leq j<r\) and \(h\in\{0,1\}^j\), define
\[
    \cH_{\mathrm{E};j, h}
    \coloneqq
    \operatorname{span}
    \left\{
        \ket {a_0, a_1, ..., a_{j-1}\,0^{r-j}}_E:
        \begin{cases}
            a_k = 0 & \text{if } h_k = 0 \\
            a_k \in \{1, ..., m\} & \text{if } h_k = 1
        \end{cases}
    \right\}.
\]
Define the corresponding valid private blocks and their direct sum by
\begin{equation}
    \cV_{{\mathrm{comp}},j,h}
    \coloneqq
    \ket j_\mathrm{T} \otimes \cH_{\mathrm{E};j, h} \otimes(\cH_{\mathrm{C}, H} \oplus \cH_{\mathrm{C}, \mathrm{D}}),
    \qquad
    \cV_{\mathrm{comp}}
    \coloneqq 
    \bigoplus_{j=0}^{r-1}\ \bigoplus_{h\in\{0,1\}^j}
       \mathcal V_{{\mathrm{comp}},j,h}.
 \label{eq:extended-valid-space}
\end{equation}

The local Hamiltonian rotation does not change the history labels, and the local dissipative rotations only change them in valid ways; thus, \(\ran\Gamma_{\mathrm{comp}}\subseteq\mathcal V_{\mathrm{comp}}\). The private component of the output from any valid private input remains in the valid private space. The oracle preserves the time and history labels. Thus, \(\mathcal V_{\mathrm{comp}}\) is invariant under \(S^\circ_{\mathrm{C} \leftarrow \mathrm{C}}\), \(O_{\mathrm{comp}}\), and their product \(S_{\mathrm{ C} \leftarrow \mathrm{C}} = S^\circ_{\mathrm{C} \leftarrow \mathrm{C}} O_{\mathrm{comp}}\). 

In the following section we use \(V_{0},V_{1}\) as defined in \cref{eq:combined-v}. 

\begin{lemma}[Extended private-path decomposition]
\label{lem:extended-private-paths}
Assume \(r\geq2 \tau\), set \(c_{\mathrm{D}}:=\sqrt{1-\delta_{\mathrm{D}}}\), and let \(P\) equal \(V_1V_1^{\dagger}\) on each dissipative private sector and zero on every Hamiltonian private sector. There are lower-triangular operators
\(\Delta_{\mathrm{R}},\Delta_{\mathrm{J}}\) on \(\mathcal V_{\mathrm{comp}}\) such that
\begin{equation}
    S^{\circ}_{\mathrm{C} \leftarrow \mathrm{C}}=i(I-P-\Delta_{\mathrm{R}}-\Delta_{\mathrm{J}}),
    \qquad
    S_{\mathrm{C} \leftarrow \mathrm{C}}=i(I-P-\Delta_{\mathrm{R}}-\Delta_{\mathrm{J}})O_{\mathrm{comp}},
    \qquad
    PO_{\mathrm{comp}}P=0.
 \label{eq:extended-private-block}
\end{equation}

The local private blocks in \cref{eq:R,eq:H-R}, together with the same-time Hamiltonian-to-dissipative path, give
\begin{equation}
    [\Delta_{\mathrm{R}}]_{j,j}
    =
    \begin{pmatrix}
        (1-c_H)P_H & 0\\
        i S^{\circ}_{\mathrm{C}\leftarrow \mathrm{P}, \mathrm{D}, j} S^{\circ}_{\mathrm{P} \leftarrow \mathrm{C}, H, j} & (1-c_{\mathrm{D}})V_0V_0^{\dagger}
    \end{pmatrix},
    \qquad
    [\Delta_{\mathrm{J}}]_{j,j}=0.
    \label{eq:extended-Delta-diagonal}
\end{equation}
and every block with \(i<j\) is zero.  For a universal \(C_0\),
\begin{align}
    \norm{[\Delta_{\mathrm{R}}]_{j,j}}&\leq C_0\delta_{\cL},\nonumber\\
    \norm{[\Delta_{\mathrm{R}}]_{i,j}}
    & \leq C_0\delta_{\cL}\,e^{-\delta_{\cL}(i-j-1)/2},&&i>j,\label{eq:extended-Delta0}\\
    \norm{[\Delta_{\mathrm{J}}]_{i,j}}
    &\leq C_0\sqrt{\delta_{\cL}}\,e^{-\delta_{\cL}(i-j-1)/2}, &&i>j.
    \label{eq:extended-Delta1}
\end{align}
Both operators preserve \(\mathcal V_{\mathrm{comp}}\). A \(\Delta_{\mathrm{R}}\) transition either leaves the history unchanged or appends a no-jump entry. By contrast, a \(\Delta_{\mathrm{J}}\) transition strictly advances time, appends one jump entry, and coherently transfers its jump index into the corresponding history register.
\end{lemma}

\begin{proof}
On a fixed time block,
\[
    S^{\circ}_{\mathrm{C} \leftarrow \mathrm{C}, H, j}=i\bigl(I-(1-c_H)P_H\bigr),\qquad
    S^{\circ}_{\mathrm{C} \leftarrow \mathrm{C}, \mathrm{D}, j}=i\bigl(I-V_1V_1^{\dagger}-(1-c_{\mathrm{D}})V_0V_0^{\dagger}\bigr).
\]
Together with \cref{eq:grouped-diagonal-block}, this proves \cref{eq:extended-Delta-diagonal} and the diagonal part of \cref{eq:extended-private-block}.

For \(i>j\), classify each private-to-public-to-private path by its departure location from the private space. \(\Delta_{\mathrm{R}}\) contains all paths that exit the private space through the Hamiltonian stage or the no-jump output of the dissipative stage. \(\Delta_{\mathrm{J}}\) contains all paths that exit the private space through the jump output of the dissipative stage.

By \cref{eq:H-R,eq:R}, Hamiltonian departures and returns have norm \(s_H\leq\sqrt{2\delta_H}\), dissipative no-jump departures and valid returns have norm \(\sqrt{\delta_{\mathrm{D}}}\), and dissipative jump departures have norm one. Valid histories exclude the norm-one path from the jump input to the private space. Thus \(\Delta_{\mathrm{R}}\) departures and all valid returns have norm \(\cO(\sqrt{\delta_{\cL}})\), while \(\Delta_{\mathrm{J}}\) departures have norm at most one.
Each path crosses \(i-j-1\) complete public-to-public steps each contributing a \(c_Hc_{\mathrm{D}}\) factor. Since there are only a constant number of departure and return stage choices,
\[
    \|[\Delta_{\mathrm{R}}]_{i,j}\|
    \leq C_0\delta_\cL(c_Hc_{\mathrm{D}})^{i-j-1},
    \qquad
    \|[\Delta_{\mathrm{J}}]_{i,j}\|
    \leq C_0\sqrt{\delta_\cL}(c_Hc_{\mathrm{D}})^{i-j-1}.
\]
Using \(c_Hc_D\leq e^{-\delta_\cL/2}\) gives the claimed bounds.

\(\Delta_{\mathrm{R}}\) paths leave the history register unchanged, while \(\Delta_{\mathrm{J}}\) paths coherently transfer their jump index into \(E_j\). These updates preserve the valid space. Chronological ordering excludes \(i<j\), and thus each off-diagonal private block equals \(-i[\Delta_{\mathrm{R}}+\Delta_{\mathrm{J}}]_{i,j}\).

The off-diagonal oracle lift gives \(V_1V_1^{\dagger}\widetilde{O}_{\widehat{L}}V_1V_1^{\dagger}=0\), and \(P\) vanishes on the Hamiltonian sector, proving \(PO_{\mathrm{comp}}P=0\).
\end{proof}

We apply the counting argument of Lemma~\ref{lem:alg-word} to the blocks \(\mathcal V_{\mathrm{comp},j,h}\). A \(\Delta_{\mathrm R}\) transition preserves the existing history and appends only zeros when it advances time, while a \(\Delta_{\mathrm J}\) transition records exactly one new jump at its departure time. Thus the same row and column counts apply.

Each block contains both Hamiltonian and dissipative private sectors and all jump-index values compatible with \(h\) so there is no dependence on \(m\). The block estimates give the corresponding fixed-word bound
\[
    \norm{W_\sigma}
    \leq C_0^{u+v}\delta_{\cL}^{u+v/2}
       \binom{r+u-1}{u}\binom{r}{v}^{1/2}.
\]

Repeating Proposition~\ref{prop:factorial-decay}, with constants enlarged to absorb \(C_0\), gives
\[
    \norm{p_{\mathrm{annih}}(S^{\circ}_{\mathrm{C}\leftarrow \mathrm{C}} O_{\mathrm{comp}})^q\big|_{\mathcal V_{\mathrm{comp}}}}
    \leq
    \left(C_{\mathrm{comp}}^*\sqrt{\frac{\tau}{q}}\right)^q
\]
whenever \(q\geq C(1+\tau)\) and \(r\geq4q\). Together with \cref{eq:extended-catalyst-bound}, the catalyst-removal error is at most
\[
    \eta_{\mathrm{comp}}:=\sqrt{2\tau}
    \left(C_{\mathrm{comp}}^*\sqrt{\frac{\tau}{q}}\right)^q.
\]

To establish the query complexity and error guarantee in \cref{thm:main-thm}, choose an integer \(q\) satisfying
\begin{equation}
    q=\Theta\bigl(c(\tau,\epsilon)\bigr)
    \label{eq:q-comp}    
\end{equation}
large enough that \(q\geq C(1+\tau)\) and \(\eta_{\mathrm{comp}}\leq\epsilon/24\), and set
\begin{equation}
    r \coloneqq \left\lceil \max\left\{4q, 2\tau,                         \frac{4C_{\mathrm{comp}}\tau^2}{\epsilon}\right\} \right\rceil.
    \label{eq:r-comp}
\end{equation}

\cref{eq:combined-global-error} and the same reuse, LCU, amplification, and channel-completion argument as above give a channel \(\Phi_{\mathrm{sim}}\) satisfying
\[
    \norm{\Phi_{\mathrm{sim}}-\exp (t \cL)}_\diamond
    \leq12\eta_{\mathrm{comp}}+C_{\mathrm{comp}}\frac{\tau^2}{r}
    \leq\epsilon.
\]

The algorithm uses \(\cO(c(\tau,\epsilon))\) queries to each of \(O_H,O_{L_1},\ldots,O_{L_m}\). This proves the query and error claims of \cref{thm:main-thm}.

\subsection{Gate Complexity}

A direct implementation of the transducer applies one local update per discretization step. This implementation has high gate complexity. We reduce the gate complexity of the transducer-based algorithm by reorganizing the query-free portion of the global transducer by the last jump event in each history string. Jump events are added to the history string from the jump branch of the dissipative evolution. In valid chronological evolution the public input to the jump branch is \(0\) and the output is completely supplied by the private state. All future evolution is the no-jump branch of dissipative evolution and Hamiltonian evolution. The no-jump branch of dissipative evolution and Hamiltonian evolution are both implemented with rotations. We compile this sequence of rotations into an efficiently implemented unitary. 

We then show that a compressed version of the history is equivalent to the full history in the reuse construction and perform segmentation to get our gate complexity bound. 

After factoring out the \(i\) phase on the private space, the jump-producing rotation in \cref{eq:R} is
\[
    \begin{pmatrix} 0 & -1 \\
        1 & 0
     \end{pmatrix}.
\]
We can perform these rotations first because the jump rotation at time \(j\) couples a private sector untouched by earlier rotations to a public history on which those rotations act trivially, so it can be moved before them.

\subsubsection{Efficient Implementation of the no-jump and Hamiltonian rotations}
Let 
\[
    c_{\mathrm{D}} = \sqrt{1-\delta_\mathrm{D}}, \qquad s_\mathrm{D} = \sqrt{\delta_{\mathrm{D}}}, \qquad c_H = \frac{1-\delta_H/2}{1+\delta_H/2}, \qquad s_H = \sqrt{1-c_H^2}, \qquad \gamma = \sqrt{1-(c_{\mathrm{D}} c_H)^2}
\] 
be scalars from the rotation matrices for the no-jump dissipative portion and Hamiltonian evolution. 

We write the block matrices for the rotations restricted to \(\cH_{\mathrm{S}} \oplus \ran V_H \oplus \ran V_0\) where \(V_H, V_0\) are isometries from \(\cH_{\mathrm{S}}\) into the Hamiltonian and no-jump private spaces respectively.

At time step \(j\) the rotation for the no-jump dissipative portion as defined in \cref{eq:R} is
\[
    R_{\mathrm D} =
    \begin{pmatrix}
        c_{\mathrm D}I_{\mathrm S} & 0 & -is_{\mathrm D}V_0^\dagger\\
        0 & I_{\ran V_H} & 0 \\
        s_{\mathrm D}V_0 & 0 & ic_{\mathrm D}V_0V_0^{\dagger}
    \end{pmatrix},\qquad
    R_H=
    \begin{pmatrix}
        c_HI_{\mathrm S} & -is_HV_H^\dagger & 0\\
        s_HV_H & ic_HP_H & 0\\
        0 & 0 & I_{\ran V_0}
    \end{pmatrix}.
\] where \(P_H = V_HV_H^{\dagger}\) is defined in \cref{eq:H-constants-proj}.

With Hamiltonian evolution applied first, their product is
\[
    R_{\mathrm D}R_H=
    \begin{pmatrix}
        c_Hc_{\mathrm{D}}I_{\mathrm S} & -ic_{\mathrm D}s_HV_H^\dagger & -is_{\mathrm D}V_0^\dagger\\
        s_HV_H & ic_HP_H & 0\\
        c_Hs_{\mathrm D}V_0 & -is_Hs_{\mathrm D}V_0V_H^\dagger & ic_{\mathrm D}V_0V_0^{\dagger}
    \end{pmatrix}.
\]

Factoring the \(i\) phase from the private states and writing the matrix under the identifications induced by \(V_H, V_0\) yields
\[
    \begin{pmatrix}
        c_Hc_{\mathrm{D}} & -s_Hc_{\mathrm{D}} & -s_{\mathrm{D}} \\
        s_{H} & c_{H} & 0 \\
        c_Hs_{\mathrm{D}} & -s_Hs_{\mathrm{D}} & c_{\mathrm{D}}
    \end{pmatrix}.
\]

We then perform a decomposition of this rotation into \(3\) unitary factors. For \(\gamma > 0\) define
\begin{align}
    F_{\mathrm{out}} &\coloneqq \frac{1}{\gamma}
    \begin{pmatrix}
        \gamma & 0 & 0 \\
        0 & s_{H} & -c_Hs_{\mathrm{D}} \\
        0 & c_H s_{\mathrm{D}} & s_H 
    \end{pmatrix} \\
    F_{\mathrm{in}} &\coloneqq \frac{1}{\gamma}
    \begin{pmatrix}
        \gamma & 0 & 0 \\
        0 & c_{\mathrm{D}} s_H & -s_{\mathrm{D}} \\
        0 & s_{\mathrm{D}} & c_{\mathrm{D}} s_H
    \end{pmatrix} \\
    R_{\mathrm{eff}} & \coloneqq 
    \begin{pmatrix}
        c_{\mathrm{D}}c_H & -\gamma & 0 \\ 
        \gamma & c_{\mathrm{D}}c_H & 0 \\
        0 & 0 & 1
    \end{pmatrix}
\end{align}
Direct multiplication gives
\[
    \begin{pmatrix}
        c_Hc_{\mathrm{D}} & -s_Hc_{\mathrm{D}} & -s_{\mathrm{D}} \\
        s_{H} & c_{H} & 0 \\
        c_Hs_{\mathrm{D}} & -s_Hs_{\mathrm{D}} & c_{\mathrm{D}}
    \end{pmatrix} = F_{\mathrm{out}}R_{\mathrm{eff}}F_{\mathrm{in}}^{\dagger}.
\]

For a last jump at time \(\ell\), let \(j \coloneqq r - 1 - \ell\) be the remaining number of time steps that undergo no-jump and Hamiltonian evolution. The remaining time steps are \(r-j, \ldots, r-1\). For the empty history take \(\ell = -1\), so \(j = r\). Let \(F_{\mathrm{out}}^{(j)}, F_{\mathrm{in}}^{(j)}\) apply \(F_{\mathrm{out}}, F_{\mathrm{in}}\) to the private spaces at these time steps. 

Let 
\[
    R_{\mathrm{eff}}^{(j)} \coloneqq R_{\mathrm{eff}, r-1} R_{\mathrm{eff},r-2} \cdots R_{\mathrm{eff}, r-j}, \qquad 
    R_{\mathrm{eff}}^{(0)} = I
\] where \(R_{\mathrm{eff}, k}\) applies \(R_{\mathrm{eff}}\) with the private component at time step \(k\).  

Let \(F_{\mathrm{out}, k}, F_{\mathrm{in}, k}\) apply \(F_{\mathrm{out}}, F_{\mathrm{in}}\) to the private space with time step \(k\).

\(F_{\mathrm{out}, k}, F_{\mathrm{in}, k}\) only act on the private space at time \(k\) so they commute with unitaries that act on other time steps. Thus the evolution after the last jump at time \(\ell\) is
\[
    F_{\mathrm{out}}^{(j)} R_{\mathrm{eff}}^{(j)} {F_{\mathrm{in}}^{(j)}}^{\dagger}.
\]

\(F_{\mathrm{out}}^{(j)}, F_{\mathrm{in}}^{(j)}\) have efficient implementations. We give an efficient implementation of \(R_{\mathrm{eff}}^{(j)}\) by decomposing it into a weighted sum of \(4\) contractions with efficient block encodings. We then implement \(R_{\mathrm{eff}}^{(j)}\) using LCU. 
\subsubsection{Efficient implementation of the effective rotation}

Let \(d=\lceil\log_2r\rceil\) and pad the time register to \(2^d\) coordinates. We relabel the suffix \(r-j,\ldots,r-1\) as \(0,\ldots,j-1\), using reversible arithmetic before and after the construction. All unused coordinates are fixed by the target unitary. All work ancillas are initialized to zero, and the encoded block is obtained by projecting them onto zero at the output. Distinct failure tests use separate flags, retained until the common success projection.
\begin{enumerate}
    \item Let \(R_{\mathrm{eff, diag}}^{(j)}\) be the diagonal portion of the effective rotation on the public and private spaces. The public amplitude must survive all \(j\) rotations. The private amplitude at any time step needs to survive \(1\) rotation and thus acquires a single \(c_Hc_{\mathrm{D}}\) factor. 
    
    Thus,
    \begin{equation}
        R_{\mathrm{eff, diag}}^{(j)} = (c_Hc_{\mathrm{D}})^j I_{\mathrm{P}} + c_Hc_{\mathrm{D}}\sum_{k<j} (\ketbra{k}{k}_{\mathrm{T}} \otimes P_{H}) + P_{\mathrm{other}}.
        \label{eq:R-eff-diag}
    \end{equation}
    where \(P_{\mathrm{other}} = I - I_{\mathrm{P}} - \sum_{k<j} (\ketbra{k}{k}_{\mathrm{T}} \otimes P_{H})\) is a projection onto the orthogonal complement of the public and active private subspaces. The projector \(P_H\) selects the private subspace on which \(R_{\mathrm{eff}}\) acts. 

    Write the suffix length as \(j=\sum_{b=0}^{d}c_b2^b\), using a \(d+1\)-qubit register. For each bit \(c_b\), apply the following rotation to a fresh ancilla, controlled on the public branch and \(c_b=1\).
    \[
        \begin{pmatrix}
            (c_Hc_{\mathrm{D}})^{2^b} & \sqrt{1 - (c_Hc_{\mathrm{D}})^{2*2^{b}}} \\
            -\sqrt{1 - (c_Hc_{\mathrm{D}})^{2*2^{b}}} & (c_Hc_{\mathrm{D}})^{2^b}
        \end{pmatrix}.
    \] 

    The combined zero amplitude is \(\underset{b:c_b = 1}{\prod} (c_Hc_{\mathrm{D}})^{2^b} = (c_Hc_{\mathrm{D}})^j\). 

    For the Hamiltonian private components with \(k < j\), we apply the following rotation to one additional ancilla
    \[
        \begin{pmatrix}
                (c_Hc_{\mathrm{D}}) & \sqrt{1 - (c_Hc_{\mathrm{D}})^{2}} \\
            -\sqrt{1 - (c_Hc_{\mathrm{D}})^{2}} & (c_Hc_{\mathrm{D}})
        \end{pmatrix}.
    \]

    The total gate cost is \(\cO(C_{\mathrm{anc}} + \log(r+1))\).

    The block-encoded contraction \(T_{\mathrm{diag}}\) satisfies \(R_{\mathrm{eff, diag}}^{(j)} = T_{\mathrm{diag}}\).
    
    \item Let \(R_{\mathrm{eff}, \mathrm{C} \leftarrow \mathrm{P}}\) be the public-to-private block of the effective rotation. For time step \(k\), the public amplitude must survive \(k\) public-to-public transitions and then transition to the private  space. 

    Thus,
    \begin{equation}
        R_{\mathrm{eff}, \mathrm{C} \leftarrow \mathrm{P}} = \gamma \sum_{k < j} (c_Hc_{\mathrm{D}})^k (\ket{k}_{\mathrm{T}} \otimes V_H).
        \label{eq:r-eff-public-to-private}
    \end{equation}

    To implement a block encoding of the map \(R_{\mathrm{eff}, \mathrm{C} \leftarrow \mathrm{P}}\), we use a unitary extension \(U_{\mathrm{prep, off}}\) satisfying 
    \[
        U_{\mathrm{prep, off}} |\psi\rangle_{\mathrm{P}} = 
        \frac{1}{\sqrt{\sum_{x=0}^{2^d-1} (c_Hc_{\mathrm{D}})^{2x}}} \sum_{k=0}^{2^d-1} (c_Hc_{\mathrm{D}})^{k} \ket k_{\mathrm{T}} \otimes V_H \ket \psi.
    \]
    The geometric time-register state is prepared by \(d\) single-qubit rotations. Together with the known embedding operations, this gives an implementation of \(U_{\mathrm{prep,off}}\) using \(\cO(C_{\mathrm{anc}}+\log(r+1))\) gates.

    We implement a unitary block encoding this map by first using one ancilla to flag inputs outside the public space, applying \(U_{\mathrm{prep, off}}\) and then using another ancilla to flag outputs outside \(\ran V_H\) or with a time label \(\geq j\).

    Reversibly computing the two flags costs \(\cO(C_{\mathrm{anc}} + \log(r+1))\) gates so the overall cost is \(\cO(C_{\mathrm{anc}} + \log(r+1))\) gates.

    Call the encoded block \(T_{\mathrm{C} \leftarrow \mathrm{P}}\). It satisfies \( R_{\mathrm{eff}, \mathrm{C} \leftarrow \mathrm{P}} = \gamma\sqrt{\sum_{k=0}^{2^d-1} (c_Hc_{\mathrm{D}})^{2k}} \ T_{\mathrm{C} \leftarrow \mathrm{P}}\).

    \item Let \(R_{\mathrm{eff}, \mathrm{P} \leftarrow \mathrm{C}}\) be the private-to-public block of the effective rotation. At time step \(k\), the private amplitude that moves to the public state must then survive an additional \(j-k-1\) public-to-public transitions. 

    Thus,
    \begin{equation}
        R_{\mathrm{eff}, \mathrm{P} \leftarrow \mathrm{C}} = -\gamma \sum_{k < j} (c_Hc_{\mathrm{D}})^{j-k-1} (\bra{k}_{\mathrm{T}} \otimes V_H^{\dagger}).
        \label{eq:r-eff-private-to-public}
    \end{equation}

    Let \(U_{\mathrm{reflect}, j}\) be the unitary such that
    \[
        U_{\mathrm{reflect}, j}\ket k = \ket {j-1-k \mod 2^d}.
    \]

    Implementing \(U_{\mathrm{reflect}, j}\) costs \(\cO(\log(r+1))\) gates. 

    To block encode the private-to-public block, append the index reflection \(U_{\mathrm{reflect}, j}\) to the public-to-private block-encoding circuit, then reverse the complete circuit, including its flags. The resulting circuit block encodes \(T_{\mathrm{P} \leftarrow \mathrm{C}} = (U_{\mathrm{reflect}, j} T_{\mathrm{C} \leftarrow \mathrm{P}})^{\dagger}\).

    \( R_{\mathrm{eff}, \mathrm{P} \leftarrow \mathrm{C}} = -\gamma\sqrt{\sum_{k=0}^{2^d-1} (c_Hc_{\mathrm{D}})^{2k}} T_{\mathrm{P} \leftarrow \mathrm{C}}\).

    \item Let \(R_{\mathrm{eff}, \mathrm{C} \leftarrow \mathrm{C}, \mathrm{off}}\) be the off-diagonal portion of the private-to-private block of the effective rotation. 

    The only path for amplitude from one private block at time step \(x\) to enter a future private block at time step \(y\) is for it to move into the public space at time step \(x\), survive \(y-x-1\) public-to-public transitions, and then move into private space at time step \(y\). It is not possible for any amplitude to move from future time steps to preceding ones. 

    Thus,
    \begin{equation}
        R_{\mathrm{eff}, \mathrm{C} \leftarrow \mathrm{C}, \mathrm{off}} = -\gamma^2 \sum_{0 \leq x < y < j} (c_Hc_{\mathrm{D}})^{y-x-1} (\ketbra{y}{x}_{\mathrm{T}} \otimes P_H ).
        \label{eq:r-eff-private-to-private-off}
    \end{equation}

    To construct the corresponding normalized block encoding, we introduce a \(d\)-qubit shift register \(\mathrm{shift}\) and choose a unitary extension \(U_{\mathrm{prep, shift}}\) satisfying 
    \[
        U_{\mathrm{prep, shift}} \ket 0_\mathrm{shift} = 
        \frac{1}{\sqrt{\sum_{z=0}^{2^d-1} (c_Hc_{\mathrm{D}})^{z}}} \sum_{k=0}^{2^d-1} (c_Hc_{\mathrm{D}})^{k/2} \ket k_\mathrm{shift}.
    \]
    Implementing this unitary takes \(\cO(\log(r+1))\) gates because it is equal to the product of single-qubit rotations.

    To implement a unitary block encoding this portion of the effective rotation we perform the following operations.
    \begin{enumerate}
        \item Use one ancilla to flag inputs outside \(\ran V_H\) or with a time label \(\geq j\). 
        \item Apply \(U_{\mathrm{prep, shift}}\) to the shift register.
        For an input time label \(x\), applying \(U_{\mathrm{prep, shift}}\) to the shift register produces the state
        \[
            \frac{1}{\sqrt{\sum_{z=0}^{2^d-1} (c_Hc_{\mathrm{D}})^{z}}} \sum_{k=0}^{2^d-1} (c_Hc_{\mathrm{D}})^{k/2} \ket k_{\mathrm{shift}} \ket x_{\mathrm{T}}.
        \] 
        \item Use another ancilla to flag if \(k+x+1 \geq j\). This flags outputs outside the valid times. 
        \item Add \(k+1\) to the time label modulo \(2^d\). For an input time label \(x\), suppressing the retained failure flags, this produces the state 
        \[
            \frac{1}{\sqrt{\sum_{z=0}^{2^d-1} (c_Hc_{\mathrm{D}})^{z}}} \sum_{k=0}^{2^d-1} (c_Hc_{\mathrm{D}})^{k/2} \ket k_{\mathrm{shift}} \ket {x+ k + 1 \mod 2^d}_{\mathrm{T}}.
        \]
        \item Apply \(U_{\mathrm{prep, shift}}^{\dagger}\) to the shift register. The amplitude for returning \(\ket k_{\mathrm{shift}}\) to \(\ket 0_{\mathrm{shift}}\) is \((c_Hc_{\mathrm{D}})^{k/2}/ \sqrt{\sum_{z=0}^{2^d-1} (c_Hc_{\mathrm{D}})^{z}}\).
    \end{enumerate}
    Implementing this unitary costs \(\cO(C_{\mathrm{anc}} + \log(r+1))\) gates. The block-encoded contraction is 
    \[
        T_{\mathrm C\leftarrow\mathrm C,\mathrm{off}}
        =
            \frac{\sum_{0\leq x<y<j} (c_Hc_{\mathrm{D}})^{y-x-1}|y\rangle\langle x|_{\mathrm T}\otimes P_H}{\sum_{z=0}^{2^d-1}(c_Hc_{\mathrm{D}})^z},
    \]

    \(R_{\mathrm{eff}, \mathrm{C} \leftarrow \mathrm{C}, \mathrm{off}} = -\gamma^2 \sum_{z=0}^{2^d-1} (c_Hc_{\mathrm{D}})^{z} \ T_{\mathrm{C} \leftarrow \mathrm{C}, \mathrm{off}}\).
\end{enumerate}

By the definitions of the contractions and \(R_{\mathrm{eff}}^{(j)}\), 
\[
    R_{\mathrm{eff}}^{(j)} = T_{\mathrm{diag}} + \gamma\sqrt{\sum_{k=0}^{2^d-1} (c_Hc_{\mathrm{D}})^{2k}} \ T_{\mathrm{C} \leftarrow \mathrm{P}} - \gamma\sqrt{\sum_{k=0}^{2^d-1} (c_Hc_{\mathrm{D}})^{2k}} \ T_{\mathrm{P} \leftarrow \mathrm{C}} - \gamma^2 \sum_{k=0}^{2^d-1} (c_Hc_{\mathrm{D}})^{k} \ T_{\mathrm{C} \leftarrow \mathrm{C}, \mathrm{off}}
\] where the three off-diagonal contractions act as zero outside of their respective source and target blocks. 

The LCU coefficients satisfy 
\begin{align*}
    \left(\gamma\sqrt{\sum_{k=0}^{2^d-1} (c_Hc_{\mathrm{D}})^{2k}}\right)^2 &= (1-(c_Hc_{\mathrm{D}})^2)\sum_{k=0}^{2^d-1} (c_Hc_{\mathrm{D}})^{2k} = 1-(c_Hc_{\mathrm{D}})^{2^{d+1}} \leq 1 \\ 
    \gamma^2 \sum_{k=0}^{2^d-1} (c_Hc_{\mathrm{D}})^{k} &= (1+c_Hc_{\mathrm{D}})(1-c_Hc_{\mathrm{D}})\sum_{k=0}^{2^d-1} (c_Hc_{\mathrm{D}})^k = (1+c_Hc_{\mathrm{D}})(1-(c_Hc_{\mathrm{D}})^{2^d}) \leq 2.
\end{align*} 
Thus the LCU normalization is at most \(5\).

We implement \(R_{\mathrm{eff}}^{(j)}\) using LCU \cite{CW12}. We add a zero-contraction term with weight chosen to make the LCU normalization equal to \(\csc(\pi/18) > 5\) and implement the operator exactly using 4 rounds of OAA \cite[Theorem~28]{GSLW19}. The total gate cost for implementing \(R_{\mathrm{eff}}^{(j)}\) is \(\cO(C_{\mathrm{anc}} + \log(r+1))\).

Implementing \(F_{\mathrm{out}}^{(j)}, {F_{\mathrm{in}}^{(j)}}^{\dagger}\) takes \(\cO( C_{\mathrm{anc}} + \log(r+1))\) gates by computing membership in the suffix and applying a controlled rotation. Thus, the total gate cost for implementing the suffix evolution \(F_{\mathrm{out}}^{(j)} R_{\mathrm{eff}}^{(j)} {F_{\mathrm{in}}^{(j)}}^{\dagger}\) is \(\cO(C_{\mathrm{anc}} +\log(r+1))\).

\subsubsection{History Compression}

In the case where the catalyst state is provided, the number of jumps in the history can be as high as the number of time steps. In the reuse construction, however, the initial private component is the zero vector, and each stage receives a fresh public input with blank history. The first application therefore produces only histories with no jumps.

In subsequent applications, the private component may be nonzero and jump events can be added to the history through the jump output of the dissipative rotation. Subsequent no-jump or Hamiltonian rotations can then transfer this amplitude carrying the new jump event back into the private space. In the chronological construction, each private sector is coupled to the public space by exactly one local rotation, and all subsequent rotations act as the identity on that sector, so this amplitude cannot produce another jump during the same global-transducer application. The history is unchanged by the oracle and the no-jump and Hamiltonian rotations. Thus, the maximum number of jump events in any history increases by at most one in each application of the global transducer. 

We encode each history as a list of time–jump-type pairs, allocating enough records for the maximum reuse length in the LCU construction, which is \(\cO(q)\). Since each transducer application adds at most one jump, every forward reuse branch agrees with the full-history construction under this encoding. Thus each reuse block \(P_iV_{\mathrm{in}}\), and hence the LCU block, are unchanged, and amplification preserves the preceding error guarantees.

With a standard reversible list implementation the size of this compressed history is \(\cO\bigl(q(\log(r+1) + \log(m+1))\bigr)\). Inserting or removing an event costs \(\cO\bigl(q(\log(r+1) + \log(m+1))\bigr)\) gates and computing the last jump costs \(\cO\bigl(q \log(r+1)\bigr)\) gates. 

\subsubsection{Efficient Implementation of the Global Transducer}
We now show how to perform an efficient implementation of the query-free portion of the global transducer and provide a gate count for this implementation. 

The compiled implementation agrees with the full-history transducer on all inputs encountered in the forward reuse circuits and is extended unitarily on unused states. The steps of the compiled implementation are as follows.
\begin{enumerate}
    \item Apply an \(i\) phase to the private state. This takes \(\cO(1)\) gates. 
    \item Apply the jump-producing rotations between the public state and the \(\ran V_1\) portion of the private state. We implement these rotations by reversibly inserting or removing the corresponding jump event from the history. This costs \(\cO\bigl(C_{\mathrm{anc}} + q(\log(r+1) + \log(m+1))\bigr)\) gates.
    \item Compute the number \(j\) of remaining time steps that undergo no-jump dissipative and Hamiltonian evolution. This takes \(\cO\bigl(q\log(r+1)\bigr)\) gates.
    \item Perform \(F_{\mathrm{out}}^{(j)} R_{\mathrm{eff}}^{(j)} {F_{\mathrm{in}}^{(j)}}^{\dagger}\) with \(j\) being the number of remaining time steps. The gate count is \(\cO(C_{\mathrm{anc}} + \log(r+1))\).
    \item Uncompute the register storing the number \(j\) of remaining time steps. This costs \(\cO\bigl(q\log(r+1)\bigr)\) gates.
\end{enumerate}

Each stacked oracle call costs \(\cO(m)\) additional gates and thus the total gate count for the global transducer is \(\cO\bigl(C_{\mathrm{anc}} + m + q(\log(r+1) + \log(m+1))\bigr)\).

In the reuse construction, we repeat the global transducer \(\cO(q)\) times. The controlled QFTs and their inverses, over all reuse lengths, require \(\cO(q^2)\) gates \cite{MZ4}. The remaining LCU operations and amplification reflections require an additional \(\cO\bigl(q^2+q(\log(r+1)+\log(m+1))\bigr)\) gates. The total gate count is \(\cO\bigl(qC_{\mathrm{anc}} + qm + q^2(\log(r+1) + \log(m+1))\bigr)\).
\subsubsection{Segmentation}
We segment our evolution into pieces of normalized time \(\tau'\).
\begin{equation}
    \beta \coloneqq \log(e + \frac{\tau+1}{\epsilon}), \qquad \tau' \coloneqq \frac{\tau}{\lceil \tau / \beta \rceil}, \qquad \eta \coloneqq \frac{\epsilon}{\lceil \tau / \beta \rceil}.
\end{equation}

With \(q, r\) as defined in \cref{eq:q-comp,eq:r-comp} and \((\tau, \epsilon)\) replaced by \((\tau', \eta)\), \(\log(r+1) = \cO(\beta), q = \cO(c(\tau', \eta)) = \cO(\beta)\).

When \(0 < \tau \leq \beta\) we use one segment. In this case the gate complexity is 
\begin{align}
    &\cO\bigl(q C_{\mathrm{anc}} + qm + q^2(\log(r+1) + \log(m+1))\bigr) \nonumber \\ 
    &= \cO\left(c(\tau, \epsilon) \left[C_{\mathrm{anc}} + m + \log(e + \frac{\tau+1}{\epsilon})\log(e +\frac{(m+1)(\tau+1)}{\epsilon})\right]\right)
\end{align}
and the query complexity is 
\begin{equation}
    \cO(c(\tau, \epsilon)).
\end{equation}

When \(\tau > \beta\), the required accuracy in each segment is \(\eta\) and the gate complexity is
\begin{align}
    &\frac{\tau}{\beta} \cO(\beta C_{\mathrm{anc}} + \beta m + \beta^3 + \beta^2\log(m+1)) \nonumber \\  &= \cO\left(c(\tau, \epsilon) \left[C_{\mathrm{anc}} + m + \log(e + \frac{\tau+1}{\epsilon})\log(e +\frac{(m+1)(\tau+1)}{\epsilon})\right]\right).
\end{align}

We discard each segment’s history and work registers before the next segment. Channel contractivity and telescoping bound the total diamond-norm error by \(\lceil \frac{\tau}{\beta} \rceil \eta=\epsilon\).

In this case the query complexity is
\begin{equation}
    \cO(\beta \lceil\frac{\tau}{\beta} \rceil) = \cO(\tau) = \cO(c(\tau, \epsilon)).
\end{equation}

This proves the gate complexity portion of the main theorem. 

\section{Extension to the time-dependent Lindblad equation}
\label{sec:time-dep}
The construction extends directly to a time-dependent generator \(\cL(t)\) with a fixed normalization bound. We assume that the time-controlled oracles block encode \(H(u)\) and \(L_i(u)\) with known normalization factors \(\alpha_H\) and \(\alpha_i\), respectively, that are independent of \(u\) and valid uniformly for \(u\in[0,t]\). Block encodings with time-dependent normalization factors can be rescaled by coherently computing the rescaling factors and applying ancilla rotations. If \(C_{\mathrm{scale}}\) denotes the total rescaling gate cost per combined SELECT invocation, the additional gate cost is \(\cO(c(\tau,\epsilon)C_{\mathrm{scale}})\), where \(\tau\) uses the fixed normalization bounds. The fixed normalization factors ensure that the scalar rotation parameters are independent of the sampling time, so the gate compilation also applies unchanged. 

Recall that \(\norm{\cL}_{\mathrm{be}} = \alpha_H + \sum_{\ell=1}^{m} \alpha_\ell^2\) and \(\tau =\norm{\cL}_{\mathrm{be}} t\). Let \(\tau_j=j\delta\) be the normalized sampling times where \(\delta=\tau/r\). For segmentation, the sampling times are shifted to the start of each segment and the discretization requirements are imposed separately on each segment. Define the normalized generator \(\cL_{\mathrm{norm}}(s) \coloneqq \frac{1}{\norm{\cL}_{\mathrm{be}}} \cL(s/\norm{\cL}_{\mathrm{be}})\). 

We assume we have access to coherently time-controlled oracles
\[
    O_X^{\mathrm{td}}
    \coloneqq
    \sum_{j=0}^{r-1}
    \bigl(\ket{j}\bra{j}\bigr)_\mathrm{T}\otimes O_X(\frac{\tau_j}{\norm{\cL}_{\mathrm{be}}}),
    \qquad X\in\{H,\widehat{L}\}.
\] We replace the time-independent oracles with their time-controlled variants. 

Because the construction already carries the time label \(j\), no other structural modification is required. The preceding simulation error estimates from the time-independent case continue to hold uniformly, since they depend only on the normalization bounds and chronological ordering.

This implements the piecewise-constant evolution
\[
    \mathcal E_r
    \coloneqq
    e^{\delta \cL_{\mathrm{norm}}(\tau_{r-1})} \circ e^{\delta \cL_{\mathrm{norm}}(\tau_{r-2})}
    \circ\cdots\circ
    e^{\delta \cL_{\mathrm{norm}}(\tau_{0})},
\]
up to the simulation error established in the time-independent case.

Define
\[
    \omega_{\cL_{\mathrm{norm}}}(\delta)
    \coloneqq
    \sup_{\substack{u,v\in[0,\tau]\\|u-v|\leq\delta}}
    \norm{\cL_{\mathrm{norm}}(u) - \cL_{\mathrm{norm}}(v)}_\diamond.
\]

Then
\[
    \norm{
        \cT\exp\left(\int_0^\tau \cL_{\mathrm{norm}}(s) \, ds\right)
        -\cE_r
    }_\diamond
    \leq \tau \omega_{\cL_{\mathrm{norm}}}(\frac{\tau}{r}),
\] where \(\cT\) denotes the time ordering operator. 

Assume that \(\lim_{a\to 0} \omega_{\cL_{\mathrm{norm}}}(a) =0\). Choose \(r\) sufficiently large to satisfy the preceding discretization requirements and
\[
    \tau\omega_{\cL_{\mathrm{norm}}} \! \left(\frac{\tau}{r}\right)
    \leq\frac{\epsilon}{4}.
\]

The additional error is then at most \(\epsilon/4\) and the preceding simulation error bounds contribute at most \(3\epsilon/4\). Thus the total error is at most \(\epsilon\). Increasing \(r\) does not increase query complexity, but it does increase the gate complexity through the dependence on \(\log(r+1)\).

For example, if \(\omega_{\cL_{\mathrm{norm}}}(\delta) \leq C\delta \), then we take \(r \geq \lceil \frac{4C\tau^2}{\epsilon} \rceil \). 

\section{AI Use}
The results of this paper were generated with the assistance of GPT 5.6-Sol and GPT 6-Astra. Interaction with GPT 5.6-Sol helped develop the main approach and assisted in developing the proofs. In particular, GPT 5.6-Sol developed the initial factorial decay proof in \cref{sec:factorial_decay}, which was later refined and significantly simplified by the authors. GPT 6-Astra developed the main idea for the efficient gate implementation and the initial proof. GPT 5.6-Sol and later GPT 6-Astra were used to check the manuscript for errors and provided revisions for correctness, clarity, and readability. The authors take sole responsibility for the correctness of the final manuscript. 
\bibliographystyle{plain}
\bibliography{refs}

\end{document}